\pdfoutput=1
\documentclass[acmsmall,10pt,screen]{acmart}
\acmJournal{PACMPL}
\acmVolume{1}
\acmNumber{POPL} 
\acmArticle{1}
\acmYear{2021}
\acmMonth{1}
\acmDOI{} 
\startPage{1}

\usepackage[utf8]{inputenc} 
\usepackage{amsmath, amsthm}
\usepackage{thmtools}
\usepackage{thm-restate}
\usepackage[english]{babel}
\usepackage{pdfpages}
\usepackage{paralist}
\usepackage{stmaryrd}
\usepackage{graphicx}
\usepackage{subcaption}
\usepackage{caption}
\usepackage{longtable}
\usepackage{booktabs}
\usepackage{tabu}
\usepackage[referable]{threeparttablex}
\usepackage{varwidth}
\usepackage{multirow}
\usepackage{wrapfig}
\usepackage{listings}
\usepackage{stackengine}
\usepackage{csquotes}
\usepackage{circuitikz}
\tikzset{>=stealth}
\ctikzset{tripoles/american and port/width=0.8, tripoles/american and port/height=.65}
\ctikzset{tripoles/american or port/width=0.8, tripoles/american or port/height=.65}

\usepackage{pifont}

\usepackage[ruled, linesnumbered]{algorithm2e}
\usepackage{parskip}
\usepackage{multicol}

\usepackage{xparse}

\usepackage{thm-restate}
\usepackage[capitalise]{cleveref}

\crefname{figure}{Figure}{Figure}
\crefformat{footnote}{#2\footnotemark[#1]#3}
\usepackage{tikz}
\usepackage{comment}

\usepackage{todonotes}
\usepackage{tikz}
\usetikzlibrary{positioning,arrows,automata,shapes,decorations,decorations.markings,calc, matrix,decorations.pathmorphing, patterns,backgrounds}

\usepackage{bm}
\usepackage{pgfplots}
\usepackage{enumerate,paralist}
\usepackage{pbox}

\usepackage{appendix}
\usepackage{calc}
\usepackage{fp}
\usepackage{multirow}
\usepackage{pbox}

\usepackage{etoolbox}

\theoremstyle{plain}
\newtheorem{remark}{Remark}

\DeclareMathAlphabet{\mathpzc}{OT1}{pzc}{m}{it}

\usepackage{graphicx} 

\mathchardef\mhyphen="2D 
\newcommand{\Nats}{\mathbb{N}}
\newcommand{\Natsplus}{\Nats^{+}}

\newcommand{\ov}{\overline}

\newcommand{\PointsToSet}[1]{\llbracket#1\rrbracket}

\newcommand{\Path}{\rightsquigarrow}
\newcommand{\DPath}[1]{\mathrel{\stackon[1pt]{$\rightsquigarrow$}{$\scriptscriptstyle#1$}}}

\newcommand{\APA}{\operatorname{APA}}

\newcommand{\SPAPA}{\operatorname{On-demand~APA}}
\newcommand{\APAPA}{\operatorname{Exhaustive~APA}}
\newcommand{\DReachability}[1]{\operatorname{D_{#1}-Reachability}}

\newcommand{\Label}{\lambda}
\newcommand{\Paragraph}[1]{\smallskip\noindent{\bf #1}}
\newcommand{\SubParagraph}[1]{\smallskip\noindent{\em #1}}
\newcommand{\Alphabet}{\Sigma}

\newcommand{\StartNonTerminal}{\mathcal{S}}

\newcommand{\OpenParenthesis}{\alpha}
\newcommand{\CloseParenthesis}{\ov{\alpha}}
\newcommand{\Otilde}{\tilde{O}}
\newcommand{\MonotoneCVP}{\operatorname{Monotone~CVP}}
\newcommand{\OV}{\operatorname{OV}}

\newcommand{\StackHeight}{\operatorname{SH}}
\newcommand{\MaxStackHeight}{\operatorname{MSH}}
\newcommand{\NumOpen}{\#_{\&}}
\newcommand{\NumClose}{\#_{*}}
\newcommand{\Dyck}{\mathcal{D}}
\newcommand{\Distance}{\delta}
\newcommand{\MaximaDistance}{\gamma}
\newcommand{\BellReachAlgo}{\operatorname{BellReachAlgo}}
\newcommand{\DOneAlgo}{\operatorname{D_{1}-ReachAlgo}}
\newcommand{\BoundedAPAAlgo}{\operatorname{BoundedAPAAlgo}}
\newcommand{\ceil}[1]{\lceil #1 \rceil}
\newcommand{\All}{\operatorname{All}}
\newcommand{\Grammar}{\mathcal{G}}

\newcommand{\Store}[1]{[#1]}
\newcommand{\Program}{\mathcal{P}}
\newcommand{\DEdge}[1]{\xrightarrow{#1}}
\newcommand{\AlgoCubic}{\operatorname{AndersenAlgo}}

\newcommand{\Worklist}{\mathcal{W}}
\newcommand{\Establish}{\operatorname{Establish}}
\newcommand{\Done}{\operatorname{Done}}
\newcommand{\ProcessEpsilon}{\operatorname{ProcessEps}}
\newcommand{\ProcessRef}{\operatorname{ProcessRef}}

\newcommand{\NC}{\operatorname{NC}}

\usetikzlibrary{arrows.meta,calc,decorations.markings,math,arrows.meta,shapes.gates.logic.US,fit}

\preto\tabular{\setcounter{magicrownumbers}{0}}
\newcounter{magicrownumbers}

\pgfdeclarelayer{bg}    
\pgfsetlayers{bg,main}  

\SetKw{Continue}{continue}

\SetCommentSty{mycommfont}

\makeatletter
\newcommand*{\centerfloat}{%
\parindent \z@
\leftskip \z@ \@plus 1fil \@minus \textwidth
\rightskip\leftskip
\parfillskip \z@skip}
\makeatother

\definecolor{mGreen}{rgb}{0,0.6,0}
\definecolor{mGray}{rgb}{0.5,0.5,0.5}
\definecolor{mPurple}{rgb}{0.58,0,0.82}
\definecolor{backgroundColour}{rgb}{0.95,0.95,0.92}
\lstdefinestyle{CStyle}{
backgroundcolor=\color{backgroundColour},   
commentstyle=\color{mGreen},
keywordstyle=\color{magenta},
numberstyle=\tiny\color{mGray},
stringstyle=\color{mPurple},
basicstyle=\footnotesize,
breakatwhitespace=false,         
breaklines=true,                 
captionpos=b,                    
keepspaces=true,                 
numbers=left,                    
numbersep=5pt,                  
showspaces=false,                
showstringspaces=false,
showtabs=false,                  
tabsize=2,
language=C
}

\newcounter{listtotal}\newcounter{listcntr}%
\NewDocumentCommand{\names}{o}{%
  \setcounter{listtotal}{0}\setcounter{listcntr}{0}%
  \renewcommand*{\do}[1]{\stepcounter{listtotal}}%
  \expandafter\docsvlist\expandafter{\namesarray}%
  \IfNoValueTF{#1}
    {\namesarray}
    {
     \renewcommand*{\do}[1]{\stepcounter{listcntr}\ifnum\value{listcntr}=#1\relax##1\fi}%
     \expandafter\docsvlist\expandafter{\namesarray}}%
}

\usepackage{booktabs}   
\usepackage{subcaption} 

\makeatletter
\g@addto@macro\bfseries{\boldmath}
\makeatother

\begin{document}

\title{The Fine-Grained and Parallel Complexity of Andersen's Pointer Analysis}

\author{Anders Alnor Mathiasen}
\affiliation{
\institution{Aarhus University}            
\streetaddress{Aabogade 34}
\city{Aarhus}
\postcode{8200}
\country{Denmark}                    
}
\email{au611509@uni.au.dk}          

\author{Andreas Pavlogiannis}
\affiliation{
\institution{Aarhus University}            
\streetaddress{Aabogade 34}
\city{Aarhus}
\postcode{8200}
\country{Denmark}                    
}
\email{pavlogiannis@cs.au.dk}          

 \begin{abstract}
Pointer analysis is one of the fundamental problems in static program analysis.
Given a set of pointers, the task is to produce a useful over-approximation of the memory locations that each pointer may point-to at runtime.
The most common formulation is Andersen's Pointer Analysis (APA), defined as an inclusion-based set of $m$ pointer constraints over a set of $n$ pointers.
Scalability is extremely important, as points-to information is a prerequisite to many other components in the static-analysis pipeline.
Existing algorithms solve APA in $O(n^2\cdot m)$ time,
while it has been conjectured that the problem has no truly sub-cubic algorithm,
with a proof so far having remained elusive.
It is also well-known that $\APA$ can be solved in $O(n^2)$ time under certain sparsity conditions that hold naturally in some settings.
Besides these simple bounds, the complexity of the problem has remained poorly understood.

In this work we draw a rich fine-grained and parallel complexity landscape of APA, and present upper and lower bounds.
First, we establish an $O(n^3)$ upper-bound for general APA, improving over $O(n^2\cdot m)$ as $n=O(m)$.
Second, we show that even \emph{on-demand} APA (``may a \emph{specific} pointer $a$ point to a \emph{specific} location $b$?'') has an $\Omega(n^3)$ (combinatorial) lower bound under standard complexity-theoretic hypotheses.
This formally establishes the long-conjectured ``cubic bottleneck'' of APA, and shows that our $O(n^3)$-time algorithm is optimal.
Third, we show that under mild restrictions, APA is solvable in $\tilde{O}(n^{\omega})$ time, where $\omega<2.373$ is the matrix-multiplication exponent.
It is believed that $\omega=2+o(1)$, in which case this bound becomes quadratic.
Fourth, we show that even under such restrictions, even the on-demand problem has an $\Omega(n^2)$ lower bound under standard complexity-theoretic hypotheses, and hence our algorithm is optimal when $\omega=2+o(1)$.
Fifth, we study the parallelizability of APA and establish lower and upper bounds:
(i)~in general, the problem is P-complete and hence unlikely parallelizable, whereas
(ii)~under mild restrictions, the problem is parallelizable.
Our theoretical treatment formalizes several insights that can lead to practical improvements in the future.
\end{abstract}



\maketitle

\section{Introduction}\label{sec:introduction}
 Programs execute by allocating memory for storing data and manipulating pointers to that memory.
Pointer analysis takes a static view of a program's heap and asks the question ``given a pointer $a$, what are the memory locations that $a$ may point-to at program runtime?''
Such information is vital to almost all questions addressed by static analyses in general~\cite{Ghiya01,Hind01},
hence many static analyzers begin with some form of pointer analysis.
In particular, for an analysis to be useful, it needs to be able to determine \emph{aliasing}, 
i.e., whether two pointers may be pointing to the same memory location.
For example, in the program of \cref{fig:intro}, the value of $c$ depends on whether $a$ and $b$ are aliases.
Naturally, aliasing is decided by (implicitly or explicitly) computing whether the intersection of the points-to space of the two pointers is empty.
\begin{figure}
\begin{subfigure}[b]{0.4\textwidth}
\centering
\begin{lstlisting}[style=CStyle]
   ...
   *a=42;
   *b=84;
   c = *a;
   //is c 42 or 84?
\end{lstlisting}
\end{subfigure}
\qquad
\qquad
\qquad
\begin{subfigure}[b]{0.4\textwidth}
\setlength\tabcolsep{5pt}
\centering
\begin{tabular}[b]{|c|c|}
\hline
Type & Statement \\
\hline
\hline
1 & $a=b$ \\
\hline
2 & $a=\&b$\\
\hline
3 &  $a=*b$\\
\hline
4 &  $*a=b$\\
\hline
\end{tabular}
\end{subfigure}
\caption{
A program where analysis depends on aliasing (left) and the four types of statements in $\APA$ (right).
}
\label{fig:intro}
\end{figure}

As usual in static analyses, points-to information can be modeled at various degrees of precision, which has consequences on the decidability and complexity of the problem.
Flow-sensitive formulations, which take into account the order of execution of pointer-manipulation statements,
are typically intractable, with results ranging from undecidability~\cite{Ramalingam94}, to PSPACE-completeness~\cite{Chakaravarthy03} and NP-/co-NP-hardness~\cite{Land91}.
In contrast, flow-insensitive formulations can be viewed as relational approaches that ignore the order of execution, and typically result in more tractable algorithmic problems.
Another feature that affects complexity is the level of indirection (i.e., how many nested dereferences can occur in a single statement)~\cite{Horwitz97}, and thus is typically kept small.
Flow-insensitive analyses are faster and achieve remarkable precision in practice~\cite{Shapiro97,Manuvir01,Blackshear11}.
This sweet spot between efficiency and precision has made flow-insensitive analyses dominant over alternatives.
Popular approaches in this domain are inclusion-based~\cite{Andersen94}, equality-based~\cite{Steensgaard96} and unification-based~\cite{Das00}.
We refer to~\cite{Smaragdakis15} for an excellent exposition.

\Paragraph{Andersen's pointer analysis.}
The most commonly used and actively studied formulation is Andersen's Pointer Analysis (APA)~\cite{Andersen94}.
The input is a set of $n$ pointers and $m$ statements of the four types shown in \cref{fig:intro}.
The solution to the analysis is the least fixpoint of a set of inclusion constraints between the points-to sets of the pointers (see \cref{subsec:apa} for details).
$\APA$ has been the subject of a truly huge body of work, ranging from 
adoptions to diverse programming languages~\cite{Sridharan09,Jang09,Lyde15},
extensions to incorporate various features (e.g., context/flow/field-sensitivity)~\cite{Whaley02,Pearce04,Hirzel04,Hardekopf11}
and implementations in various frameworks~\cite{Lhotak03,Soot,Wala}, to name a few.

\Paragraph{Complexity.}
The standard statement in the literature with regards to the complexity of $\APA$ is that it is cubic.
However, the parameter ($n$ or $m$, for $n$ pointers and $m$ statements) on which this cubic bound is expressed is often left unspecified, leading to a variety of statements.
The standard expression is an $O(m^3)$ bound~\cite{Melski00,SPA}, by reducing the problem to $m$ inclusion set constraints~\cite{McAllester99}.
Other works give a more refined bound of $O(n^2\cdot m)$~\cite{Pearce04,Kodumal04} which is an improvement over $O(m^3)$ as $n=O(m)$.
Note that, in general, $m$ can be as large as $\Theta(n^2)$, hence both types of statements result in worst-case dependency on $n$ that is at least quartic, as already observed in~\cite{Kodumal04}.
Finally, in most literature, the core algorithm constructs incrementally the closure of a flow graph by introducing edges dynamically.
However, the complexity analysis often ignores the cost for inserting edges.
As already noted by others~\cite{Heintze97,Sridharan09}, the cost of edge insertion needs to be accounted for when analyzing the complexity.
Under this consideration, the complexity of all these approaches is at least quartic in $n$.

The need for a faster algorithm is apparent from the long literature of heuristics~\cite{Rountev00,Su00,Heintze01,Berndl03,Pearce04,Hardekopf07,Xu09,Fahndrich98,Aiken97,Pek14,Dietrich15,Vedurada19}.
Despite all efforts, no algorithmic breakthrough below the cubic bound has been made for over 25 years.
In some cases, the complexity of $\APA$ can be reduced to quadratic~\cite{Sridharan09}.
This reduction holds when the instances adhere to certain sparsity conditions, which hold naturally in some settings.

\Paragraph{Exhaustive vs on-demand.}
One popular approach to reducing running time lies on the observation that we are typically interested in \emph{on-demand} variants of the problem.
That is, we want to decide whether $a$ may point to $b$ for a \emph{given} pointer $a$ and memory location $b$,
rather than the \emph{exhaustive} case that computes the points-to set of every pointer.
Under this restriction, many techniques devise analysis algorithms that aim to solve on-demand $\APA$ faster~\cite{Heintze01b,Sridharan05,Zheng08,Lu13,Zhang13,Sui16,Chatterjee18,Vedurada19}.
On the theoretical side, it is an open question whether on-demand analysis has lower complexity than exhaustive analysis.

\Paragraph{Lower bounds and cubic bottlenecks.}
Despite the complete lack of algorithmic improvements for $\APA$ for over 25 years, no lower bounds are known.
The two basic observations for exhaustive $\APA$ are that
(i)~the output has size $\Theta(n^2)$, which leads to a trivial similar lower bound for running time, and
(ii)~the problem is at least as hard as computing the transitive closure of a graph~\cite{Sridharan09,Zhang2020}.
For the on-demand case, no lower bound is known.
$\APA$ is often reduced to  a specific framework of set constraints~\cite{Heintze92,Su00}, which is computationally equivalent to CFL-Reachability~\cite{Melski00}.
Set constraints and CFL-Reachability are known to have cubic lower bounds~\cite{Heintze97}.
Unfortunately, these lower bounds do not imply any lower bound for $\APA$ as the reduction is only one way (i.e., \emph{from} $\APA$ \emph{to} set constraints).
The recurrent encounter of cubic complexity is frequently referred to as the ``cubic bottleneck in static analysis'',
though the bottleneck is only conjectured for $\APA$, with a proof so far having remained elusive.

\Paragraph{Parallelization.}
The demand for high-performance static analyses has lead various researchers to implement parallel solvers of $\APA$~\cite{MendezLojo2010,Mendez12,Su14,Wang17,Liu19,Blass19}.
Despite their practical performance, the parallelizability of $\APA$ has remained open on the theoretical level.
In contrast, the richer problem of set constraints is known to be non-parallelizable, via its reduction to CFL-Reachability~\cite{Melski00} which is known to be P-complete~\cite{Reps96}.

\subsection{Our Contributions}\label{subsec:contributions}
In this work, we draw a rich fine-grained complexity landscape of Andersen's Pointer Analysis,
by resolving open questions and improving existing bounds.
We refer to \cref{sec:summary} for a formal presentation of our main results as well as a discussion on their implications to the theory and practice of pointer analysis.

\Paragraph{Main contributions.}
Consider as input an $\APA$ instance $(A,S)$ of $n=|A|$ pointers and $m=|S|$ statements.
Our main contributions are as follows.
\begin{compactenum}
\item We show that $\APAPA$ is solvable in $O(n^3)$ time, regardless of $m$.
To our knowledge, this gives the sharpest cubic bound on the analysis, as it holds even when $m=\Theta(n^2)$.
\item We show that even $\SPAPA$ does not have a (combinatorial) sub-cubic algorithm (i.e., with complexity $O(n^{3-\epsilon}$ for some fixed $\epsilon>0$) under the combinatorial BMM hypothesis.
This formally proves the long-conjectured cubic bottleneck for $\APA$.
\item We consider a bounded version of $\APAPA$ where points-to information is witnessed by bounding the execution of type-4 (\cref{fig:intro}) statements by a poly-logarithmic bound.
We show that bounded $\APAPA$ is solvable in $\Otilde(n^{\omega})$ time, where $\Otilde$ hides poly-logarithmic factors and $\omega$ is the matrix-multiplication exponent.
It is known that $\omega< 2.373$~\cite{LeGall14}, hence our algorithm is sub-cubic.
\item It is believed that $\omega=2+o(1)$, in which case our previous bound becomes nearly quadratic.
We complement this result by showing that even $\SPAPA$ with witnesses that are logarithmically bounded (i.e., a simpler problem than that in Item~3) does not have a sub-quadratic algorithm (i.e., with complexity $O(n^{2-\epsilon})$ for some fixed $\epsilon>0$) under the Orthogonal Vectors hypothesis~\cite{Williams19}.
Hence, our algorithm for Item~3 is optimal when $\omega=2+o(1)$.
\item We show that $\APA$ is P-complete, and hence unlikely parallelizable.
On the other hand, we show that bounding $\APA$ as in Item~4 is in NC, and hence highly parallelizable.
\end{compactenum}

\SubParagraph{Technical contributions.}
Our main theoretical results rely on a number of technical novelties that might be of independent interest.
\begin{compactenum}
\item Virtually all existing algorithms for $\APA$ represent the analysis as  a \emph{flow-graph} that captures inclusion constraints between pointers.
In contrast, we develop a \emph{Dyck-graph} representation over the Dyck language of 1 parenthesis type $\Dyck_1$, which allows us to develop new insights for the problem, and establish upper and lower bounds.
\item We show that Dyck-Reachability over $\Dyck_1$ can be solved in time $\Otilde(n^{\omega})$, where $\omega$ is the matrix-multiplication exponent,
by a purely combinatorial reduction of the problem to $O(\log^2 n)$ matrix-multiplications.
\item Our lower bounds are based on \emph{fine-grained complexity}, an emerging field in complexity theory that establishes relationships between problems in P.
We believe that this field can have an important role in understanding and optimizing static program analyses.
Our work makes some of the first steps in this direction.
\end{compactenum}

\section{Preliminaries}\label{sec:prelim}

In this section we give a formal presentation of Andersen's pointer analysis and develop some general notation.
We also define Dyck graphs and show how reachability relationships in such graphs can be used to represent points-to relationships between pointers.
Finally, we present the main theorems of this paper.
Given a number $n\in \Nats$, we denote by $[n]$ the set $\{1,2,\dots, n\}$.


\subsection{Andersen's Pointer Analysis}\label{subsec:apa}

We begin with giving the formal definition of Andersen's pointer analysis, as well as a bounded version of the problem.

\Paragraph{Andersen's pointer analysis ($\APA$).}
An instance of $\APA$ is a pair $(A,S)$, where $A$ is a set of $n$ \emph{pointers}\footnote{Although in practice not all variables are pointers, we will use this term liberally for simplicity of presentation.} and $S$ is a set of $m$ \emph{statements}.
Each statement has one of the four types shown in \cref{tab:apa_rules}.
Conceptually, the pointers may reference memory locations during the runtime of a program which uses the statements to manipulate the pointers.
\begin{compactenum}
\item A type~1 statement $a=b$ represents pointer assignment.
\item A type~2 statement $a=\&b$ represents making $a$ point to the location of $b$.
\item A type~3 statement $a=*b$ represents an indirect assignment $a=c$, where $c$ is pointed by $b$.
\item A type~4 statement $*a=b$ represents an indirect assignment $c=b$, where $c$ is pointed by $a$.
\end{compactenum}
\begin{table}
\setlength\tabcolsep{5pt}
\begin{tabular}{|c||c|c|c|}
\hline
Type & Statement & Inclusion Constraint & Operational Semantics \\
\hline
\hline
1 & $a=b$ & $\PointsToSet{b}\subseteq \PointsToSet{a}$ & $\Store{a}\gets \Store{a} \cup \Store{b}$\\
\hline
2 & $a=\&b$ & $b\in \PointsToSet{a}$ & $\Store{a}\gets \Store{a}\cup \{ b \}$ \\
\hline
3 &  $a=*b$ & $\forall c\in \PointsToSet{b}\colon \PointsToSet{c}\subseteq \PointsToSet{a}$ & $\Store{a}\gets \Store{a} \cup \left(\bigcup_{c\in \Store{b}} \Store{c}\right)$\\
\hline
4 &  $*a=b$ & $\forall c \in \PointsToSet{a}\colon \PointsToSet{b}\subseteq \PointsToSet{c}$ & $\forall c\in \Store{a}\colon \Store{c}\gets \Store{c}\cup \Store{b}$ \\
\hline
\end{tabular}
\caption{
The four types of statements of $\APA$, the inclusion constraints they generate and the associated operational semantics.
}
\label{tab:apa_rules}
\end{table}
As standard practice, more complex statements such as $*a=*b$ have been normalized by introducing slack pointers~\cite{Andersen94,Horwitz97,SPA}.
We assume wlog that $S$ does not contain the same statement twice, and hence $m=O(n^2)$
\footnote{Note that the analyzed program might indeed contain the same statement twice. Generating the $\APA$ instance $(A,S)$ from the program is performed in time linear in the size of the program, after which the input is in the assumed form.}.
Given some $i\in [4]$, we denote by $S_i$ the statements of $S$ of type $i$.
Given a pointer $a$, we let $\PointsToSet{a}\subseteq A$ be the \emph{points-to} set of $a$.
Typically in pointer analysis, $\PointsToSet{a}$ is an over-approximation of the locations that $a$ can point-to during the lifetime of a program.
In Andersen's inclusion-based pointer analysis, the sets $\PointsToSet{a}$ are defined as follows.
Each statement generates an inclusion constraint between various points-to sets, as shown  in \cref{tab:apa_rules}.
The solution to $\APA$ is the smallest assignment $\{\PointsToSet{a}\to 2^{A}\}_{a\in A}$ that satisfies all constraints.

\SubParagraph{Exhaustive vs on-demand.}
As standard in the literature, we distinguish between the \emph{exhaustive} and \emph{on-demand} versions of the problem.
In each case, the input is an instance of $\APA$.
The $\APAPA$ problem asks to compute the points-to set of every pointer $a\in A$.
The $\SPAPA$ problem asks to compute whether $b\in \PointsToSet{a}$ for a given pair of pointers $a,b\in A$.
Hence, $\SPAPA$ is a simplification of $\APAPA$ where the size of the output is a single bit, as opposed to $\Theta(n^2)$ bits required to output the points-to set of every pointer.
The two variants can be viewed as analogues to the all-pairs and single-pair formulations of graph problems (e.g., reachability).

\SubParagraph{Operational semantics.}
Since the statements in $\APA$ come out of programs, it is convenient to consider them as executable instructions and assign simple operational semantics to them.
The semantics are over a global store $\Store{}\colon A\to 2^{A}$ that maps every pointer to its points-to set, which is initially empty.
Executing one statement corresponds to updating the store as shown in \cref{tab:apa_rules}.
This operational view already hints an (albeit inefficient) algorithm for solving $\APA$, namely, by iteratively executing some statement until no execution modifies the store.

\begin{figure*}
\centerline{
\begin{subfigure}[b]{0.25\textwidth}
\setlength\tabcolsep{7pt}
\begin{tabular}{ll}
$b=\& a$, & $d=\&a$\\
$c=\&b$, & $d=\& c$\\
$*d=c$, & $d=*e$ \\
$e=*d$, & $e=\&f$ \\
\end{tabular}
\end{subfigure}
\quad
\begin{subfigure}[b]{0.28\textwidth}
\setlength\tabcolsep{7pt}
\begin{tabular}{ll}
1.~$b=\&a$ & 5.~$d=*e$\\
2.~$c=\&b $ & 6.~$*d=c$\\
3.~$d=\&c$ & 7.~$ d=a$\\
4.~$e=*d$ & 8.~$ *d=c $\\
\end{tabular}
\end{subfigure}
\quad
\begin{subfigure}[b]{0.3\textwidth}
\setlength\tabcolsep{7pt}
\begin{tabular}{ll}
1.~$\Store{b}\gets\{a\}$ & 5.~$\Store{d}\gets\{c,a\}$\\
2.~$\Store{c}\gets\{b\} $ & 6.~$\Store{a}\gets\{b\}$\\
3.~$\Store{d}\gets\{c\}$ & 7.~$ \Store{d}\gets\{ c,a,b \} $\\
4.~$\Store{e}\gets\{b\}$ & 8.~$ \Store{b}\gets\{a,b\} $\\
\end{tabular}
\end{subfigure}
}
\caption{
An instance of $\APA$ (left), a witness program for $b\in \PointsToSet{b}$ (middle),
and the updates to the store while executing the witness (right).
Note that the statement $e=\&f$ is not used, while $*d=c$ is used twice.
}
\label{fig:apa_example}
\end{figure*}

\SubParagraph{Witnesses.}
The operational semantics allow us to define witnesses of points-to relations.
Given two pointers $a,b\in A$, a \emph{witness program} (or simply, \emph{witness}) for $b\in\PointsToSet{a}$
is a sequence of statements from $S$ that results in $b\in \Store{a}$.
See \cref{fig:apa_example} for an illustration.

\Paragraph{Bounded $\APA$.}
Motivated by practical applications, we introduce a bounded version of $\APA$ that restricts the length of witnesses.
Consider two pointers $a,b\in A$ and a witness program $\Program$ for $b\in \PointsToSet{a}$.
Given some $i\in[4]$ and $j\in \Nats$, we say that $\Program$ is $(i,j)$\emph{-bounded}
if $\Program$ executes at most $j$ statements of type~$i$.
For example, the witness for $b\in \PointsToSet{b} $ in \cref{fig:apa_example} is $(3,2)$-bounded but not $(2,2)$-bounded.
The problem of $(i,j)$-bounded $\APA$ asks for a solution to the inclusion constraints of $\APA$
such that for any two pointers $a,b$ with $b\not \in \PointsToSet{a}$, any witness program that results in $b\in \Store{a}$ executes more than $j$ statements of type~$i$.

\smallskip
\begin{remark}\label{rem:bounded}
The boundedness of $(i,j)$-bounded $\APA$ is only one-way, i.e., for relationships of the form $b\not \in \PointsToSet{a}$ and not of the form $b \in \PointsToSet{a}$.
In particular, it is allowed to have $b\in \PointsToSet{a}$ even if this is witnessed only by programs that execute statements of type~$i$ more than $j$ times.
\end{remark}

Bounded versions of $\APA$ do not necessarily have a unique solution, e.g., if the shortest witness for $a$ pointing to $b$ exceeds the bound, we can have $b\in\PointsToSet{a}$ or $b\not \in \PointsToSet{a}$.
However, any solution suffices as long as 
(i)~every points-to relationship $b\in\PointsToSet{a}$ reported has a witness, and
(ii)~all points-to relationships that have a bounded witness are reported (wrt the given bound).
Similar techniques for witness bounding are used widely in practice in order to speed up static analyses.

\subsection{Dyck Reachability and Representation of Andersen's Pointer Analysis}\label{subsec:dyck_reachability}

Here we develop some notation on Dyck languages and Dyck reachability,
and use it to represent instances of $\APA$ as Dyck graphs.

\Paragraph{Dyck languages.}
Given a non-negative integer $k\in \Nats$, we denote by $\Alphabet_k=\{\epsilon\}\cup \{\OpenParenthesis_i, \CloseParenthesis_i\}_{i=1}^k$
a finite \emph{alphabet} of $k$ parenthesis types, together with a null element $\epsilon$.
We denote by $\Dyck_k$ the Dyck language over $\Alphabet_k$,
defined as the language of strings generated by the following context-free grammar $\Grammar_k$:
\[
\StartNonTerminal \to \StartNonTerminal ~ \StartNonTerminal ~ | ~ \OpenParenthesis_1 ~ \StartNonTerminal \CloseParenthesis_1 ~ | ~ \cdots ~ | ~ \OpenParenthesis_k ~ \StartNonTerminal \CloseParenthesis_k ~ | ~\epsilon
\]
In words, $\Dyck_k$ contains all strings where parentheses are properly balanced.
In this work we focus on the special case where $k=1$, i.e., we have only one parenthesis type.
To capture the relationship between $\Dyck_1$ and $\APA$, we will let $\OpenParenthesis_1$ be $\&$ and $\CloseParenthesis_1$ be $*$.
This relationship will become clearer later in this section.

\Paragraph{Dyck graphs.}
A Dyck graph $G=(V, E)$ is a digraph where edges are labeled with elements of $\Alphabet_1$, i.e.,
$E\subseteq V\times V\times \Alphabet_1$ and edges have the form $\tau=(a,b,\Label)$.
Often we will only be interested in the endpoints $a,b$ of an edge, in which case we represent $\tau=(a,b)$,
and we will denote by $\Label(\tau)$ the label of $\tau$.
The label $\Label(P)$ of a path in $G$ is the concatenation of the labels along the edges of $P$.
We often represent $P$ from $a$ to $b$ graphically as $a\DPath{\Label(P)}b$,
and, given some $i\in \Nats$, we write $a\DPath{i\& }b$ (resp.,$a\DPath{i* }b$ ) to denote a path $P\colon a\DPath{}b$ 
with label $i$ consecutive symbols $\&$ (resp., $*$), possibly interleaved with $\epsilon$ symbols.
We say that $b$ is \emph{Dyck-reachable} (or \emph{D-reachable}) from $a$ if there exists a path $P\colon x\Path y$ with $\Label(P)\in \Dyck_1$.
We say that $b$ \emph{flows into} $a$ via a node $c$ if 
(i)~we have $b\DEdge{\&}c$, and 
(ii)~$a$ is D-reachable from $c$.
The $\DReachability{1}$ problem takes as input a Dyck graph and asks to return all pairs of nodes $(b,a)$ such that $a$ is D-reachable from $b$.

\Paragraph{Graph representation of $\APA$.}
For convenience, we frequently represent instances of $\APA$ using Dyck graphs.
In particular, given an instance $(A,S)$ of $\APA$, we use a Dyck graph $G=(A,E)$ where $E$ represents all statements in $S\setminus S_4$.
In particular, we have the following edges.
\begin{compactenum}
\item For every type~1 statement $a=b$, we have $b\DEdge{\epsilon} a$ in $E$.
\item For every type~2 statement $a=\&b$, we have $b\DEdge{\&} a$ in $E$.
\item For every type~3 statement $a=*b$, we have $b\DEdge{*} a$ in $E$.
\end{compactenum}
The instance $(A,S)$ is represented as a pair $(G,S_4)$ where $G=(A,E)$ is the Dyck graph and $S_4$ is the type-4 statements of $S$.
See \cref{fig:dyck} for an illustration.
The motivation behind this representation comes from the following lemma,
which establishes a correspondence between paths in Dyck graphs and $\APA$ without statements of type~4.

\smallskip
\begin{restatable}{lemma}{lemapatodyck}\label{lem:apa_to_dyck}
Consider the Dyck graph representation $(G=(E,V),S_4)$ of $(A,S)$, and the modified $\APA$ instance $(A,S\setminus S_4)$.
For every two pointers $a, b\in A$, we have $b\in \PointsToSet{a}$ in $(A, S\setminus S_4)$ iff $b$ flows into $a$ in $G$.
\end{restatable}

\Paragraph{Resolved Dyck graphs.}
Consider an $\APA$ instance $(A,S)$ and the corresponding Dyck-graph representation $(G=(A,E), S_4)$.
The \emph{resolved Dyck graph} $\ov{G}=(A, \ov{E})$ is a Dyck graph where $\ov{E}$ is the smallest set that satisfies the following conditions.
\begin{compactenum}
\item $E\subseteq \ov{E}$.
\item For every statement $*a=b$, for every node $c$ that flows into $a$ in $\ov{G}$, we have $b\DEdge{\epsilon}c$.
\end{compactenum}
Intuitively, all type~4 statements have been resolved as $\epsilon$-edges in $\ov{G}$.
See \cref{fig:dyck} for an illustration.
The following lemma follows directly from \cref{lem:apa_to_dyck}.

\smallskip
\begin{restatable}{lemma}{lemresolvedgraph}\label{lem:resolved_graph}
For every two pointers $a, b\in A$, we have $b\in \PointsToSet{a}$ iff $b$ flows into $a$ in $\ov{G}$.
\end{restatable}

\Paragraph{Intuition behind the Dyck-graph representation.}
Virtually all algorithms for $\APA$ in the literature use a flow graph for representing inclusion relationships between pointers.
Pointer inclusion occurs when the analysis discovers that, for two pointers $a,b$ we have $\PointsToSet{b}\subseteq \PointsToSet{a}$, represented via an edge $b\to a$ in the flow graph.

Our Dyck graph $G$ is a richer structure compared to the standard flow graph.
In fact, we obtain the (initial) flow graph if we remove from $G$ edges representing pointer references (labeled with $\&$) and dereferences of type~3 (labeled with $*$).
The only information missing from $G$ is statements of type~4.
The analysis can be seen as iteratively discovering type~4 statements $*a=b$ and inserting an edge $b\DEdge{\epsilon}c$ in $G$.
This process terminates with the resolved graph $\ov{G}$.
\cref{lem:resolved_graph} implies that, at that point, all points-to information can be expressed as flows-into relationships in $\ov{G}$.

\begin{figure}
\newcommand{\xdisposition}{3.6}
\newcommand{\ydisposition}{0}
\newcommand{\xtstep}{0.9}
\newcommand{\ytstep}{0.5}
\newcommand{\xstep}{1.1}
\newcommand{\ystep}{0.9}
\newcommand{\gatex}{0.8}
\newcommand{\gatey}{0.4}

\centering
\begin{tikzpicture}[thick,
pre/.style={<-, thick, shorten >=0cm,shorten <=0cm},
post/.style={->, thick, shorten >=-0cm,shorten <=-0cm},
seqtrace/.style={->, line width=2},
postpath/.style={->, thick, decorate, decoration={zigzag,amplitude=1pt,segment length=2mm,pre=lineto,pre length=2pt, post=lineto, post length=4pt}},
node/.style={very thick, draw=black, circle, inner sep = 1, minimum size=5.75mm},
]

\begin{scope}[shift={(0*\xstep,0*\ystep)}]

\node[] at (2.5*\xstep, 1.5*\ystep) {$G$};

\node[] at(0*\xstep, 1*\ystep)	{
$\begin{aligned}
*d=c\\
\end{aligned}$
};

\node[node] (a) at (0*\xstep, 0*\ystep){$a$};
\node[node] (b) at (1*\xstep, 0*\ystep){$b$};
\node[node] (c) at (2*\xstep, 0*\ystep){$c$};
\node[node] (d) at (3*\xstep, 0*\ystep){$d$};
\node[node] (e) at (4*\xstep, 0*\ystep){$e$};
\node[node] (f) at (5*\xstep, 0*\ystep){$f$};

\draw[post] (a) to node[above, midway]{$\&$} (b);
\draw[post] (b) to node[above, midway]{$\&$} (c);
\draw[post] (c) to node[above, midway]{$\&$} (d);
\draw[post, bend left=20] (d) to node[above, midway]{$*$} (e);
\draw[post, bend left=20] (e) to node[below, midway]{$*$} (d);
\draw[post,] (f) to node[above, midway]{$\&$} (e);
\draw[post, bend right=30] (a) to node[below, midway]{$\&$} (d);

\end{scope}

\begin{scope}[shift={(6.75*\xstep,0*\ystep)}]

\node[] at (2.5*\xstep, 1.5*\ystep) {$\ov{G}$};

\node[node] (a) at (0*\xstep, 0*\ystep){$a$};
\node[node] (b) at (1*\xstep, 0*\ystep){$b$};
\node[node] (c) at (2*\xstep, 0*\ystep){$c$};
\node[node] (d) at (3*\xstep, 0*\ystep){$d$};
\node[node] (e) at (4*\xstep, 0*\ystep){$e$};
\node[node] (f) at (5*\xstep, 0*\ystep){$f$};

\draw[post] (a) to node[above, midway]{$\&$} (b);
\draw[post] (b) to node[above, midway]{$\&$} (c);
\draw[post] (c) to node[above, midway]{$\&$} (d);
\draw[post, bend left=20] (d) to node[above, midway]{$*$} (e);
\draw[post, bend left=20] (e) to node[below, midway]{$*$} (d);
\draw[post,] (f) to node[above, midway]{$\&$} (e);
\draw[post, bend right=30] (a) to node[below, midway]{$\&$} (d);

\draw[post, out=110, in =70, looseness=1.2] (c) to node[above, midway]{$\epsilon$} (a);
\draw[post, out=120, in =60, looseness=1] (c) to node[above, midway]{$\epsilon$} (b);

\draw[post, loop, out=80, in=100, looseness=15] (c) to node[above, midway]{$\epsilon$} (c);

\end{scope}

\end{tikzpicture}
\caption{
The Dyck graph for the $\APA$ instance of \cref{fig:apa_example} (left), and the resolved Dyck graph $\ov{G}$ (right).
}
\label{fig:dyck}
\end{figure}
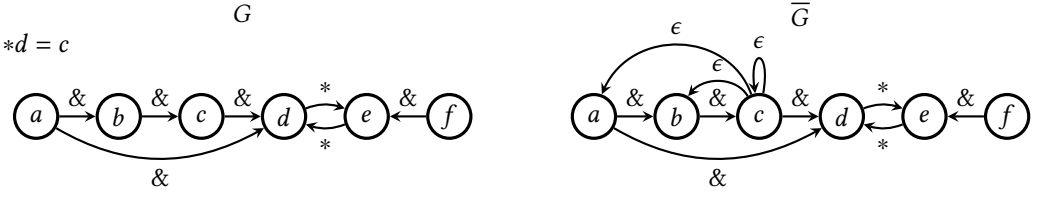

\section{Summary of Main Results}\label{sec:summary}

We are now ready to present our main theorems, followed by a discussion on their implications to the theory and practice of pointer analysis.
In later sections we develop the proofs. 

\Paragraph{Cubic upper-bound of $\APA$.}
It is well stated that $\APA$ can be solved in cubic time.
However, ``cubic'' refers to the size of the input, and typically has the form $O(m^3)$~\cite{Melski00,SPA} or $O(n^2\cdot m)$~\cite{Pearce04,Kodumal04},  for $n$ pointers and $m$ statements.
Note that $m$ can be as large as $\Theta(n^2)$, which yields the bound $O(n^4)$, as already observed in~\cite{Kodumal04}.
Our first theorem shows that in fact, the problem is solvable in $O(n^3)$ time regardless of $m$.
Although we do not consider this our major result, we are not aware of a proven $O(n^3)$ bound with $n$ being the number of pointers.
We also hope that the theorem will provide a future reference for a formal $O(n^3)$ complexity statement for $\APA$.
\smallskip
\begin{restatable}{theorem}{themcubicupperbound}\label{them:cubic_upper_bound}
$\APAPA$ is solvable in $O(n^3)$ time, for any $m$, where $n$ is the number of pointers and $m$ is the number of statements.
\end{restatable}

%
%
%

\Paragraph{Cubic hardness of $\APA$.}
Given the cubic upper-bound of \cref{them:cubic_upper_bound}, it is natural to ask whether sub-cubic algorithms exist for the problem,
i.e., algorithms with running time $O(n^{3-\epsilon})$, for some fixed $\epsilon >0$.
Indeed, the rich literature of heuristics (e.g.,~\cite{Rountev00,Su00,Heintze01,Berndl03,Pearce04,Hardekopf07,Xu09,Fahndrich98,Aiken97,Pek14,Dietrich15,Vedurada19}) is indicative of the need for such an improvement.
On the other hand, no lower-bound has been known.
In fine-grained complexity, there is a widespread distinction between \emph{combinatorial} and \emph{algebraic} algorithms.
The most famous combinatorial lower bound is for Boolean Matrix Multiplication (BMM).
The respective hypothesis states that there is no combinatorial $O(n^{3-\epsilon})$ algorithm for multiplying two $n\times n$ Boolean matrices, for any fixed $\epsilon>0$.
The BMM hypothesis has formed the basis for many lower bounds in graph algorithms, verification, and static analysis~\cite{Bansal09,Williams18b,AbboudW14,Chatterjee16b,Chatterjee18}.

Given the combinatorial nature of $\APA$, we examine whether combinatorial sub-cubic improvements are possible.
First, note that the edge set $\{(x_i,x_j)\}$ of a digraph can be represented as a set of pointer assignments $\{(x_j=x_i)\}$.
This observation leads us to the following remark.
\smallskip
\begin{remark}\label{rem:bmm_hard}
Even without statements of type~3 and type~4 (i.e., without pointer dereferences),
$\APAPA$ is at least as hard as computing all-pairs reachability in a digraph.
Since, under the BMM hypothesis, 
all-pairs reachability does not have a combinatorial sub-cubic algorithm,
the same lower-bound follows for $\APAPA$.
\end{remark}

On the other hand, single-pair reachability is solvable in linear-time in the size of the graph,
and is thus considerably easier than its all-pairs version.
Hence, the relevant question is whether $\SPAPA$ (i.e., given pointers $a,b$, is it the case that $b\in\PointsToSet{a}$?) has sub-cubic complexity.
We answer this question in negative.
\smallskip
\begin{restatable}{theorem}{thembmmhard}\label{them:bmm_hard}
$\SPAPA$ has no sub-cubic combinatorial algorithm under the BMM hypothesis.
\end{restatable}

Note that \cref{them:bmm_hard} indeed relates a problem with output size $\Theta(1)$ ($\SPAPA$) to a problem with output size $\Theta(n^2)$ (BMM).
The theorem has four implications.
First, it establishes formally the long-conjectured ``cubic bottleneck'' for $\APA$.
Second, it shows that the algorithm of \cref{them:cubic_upper_bound} is optimal even for $\SPAPA$, as far as combinatorial algorithms are concerned.
Third, it indicates that the hardness of $\APA$ does not come from the requirement to produce large-sized outputs (i.e., of size $\Theta(n^2)$ for the points-to set of each pointer), as sub-cubic complexity is also unlikely for constant-size outputs.
Fourth, it shows that all on-demand analyses (e.g.,~\cite{Heintze01b,Sridharan05,Zheng08,Lu13,Zhang13,Chatterjee18,Vedurada19}),
which attempt to reduce complexity by avoiding the exhaustive computation of all points-to sets, can only provide heuristic improvements without any guarantees.

\Paragraph{Bounded $\APA$.}
Given the hardness of $\APA$ under \cref{them:bmm_hard}, we next seek  mild restrictions that allow for algorithmic improvements below the cubic bound.
Perhaps surprisingly, we show that bounding the number of times statements of type~4 are executed suffices.
\smallskip
\begin{restatable}{theorem}{themboundedapasubcubic}\label{them:bounded_apa_subcubic}
For all $j\in \Natsplus$,  $(4,j)$-bounded $\APAPA$ is solvable in $\Otilde(n^{\omega}\cdot j)$ time, where $n$ is the number of pointers and $\omega$ is the matrix-multiplication exponent.
In particular, $(4,\Otilde(1))$-bounded $\APAPA$ is solvable in $\Otilde(n^{\omega})$ time.
\end{restatable}

Here $\Otilde$ hides poly-logarithmic factors (i.e., factors of the form $\log^{c} n$, for some constant $c$).
It is known that $\omega<2.373$~\cite{LeGall14}, hence the bound is sub-cubic.
Besides its theoretical interest, \cref{them:bounded_apa_subcubic} also has practical relevance, as
it reduces the problem to  a small number of matrix multiplications.
First, some sub-cubic algorithms for matrix multiplication, like Strassen's~\cite{Strassen69}, often lead to observable practical speedups over simpler, cubic algorithms~\cite{Huang16,Lederman96}.
Second, this reduction can take advantage of both highly optimized practical implementations for matrix multiplication~\cite{Ladernam92,Kaporin99} and specialized hardware~\cite{Dave07}.
Third, the algorithm behind \cref{them:bounded_apa_subcubic} is an ``anytime algorithm'', in the spirit of~\cite{Boddy91,Chatterjee15a}.
The algorithm computes $(4,j)$-bounded $\APAPA$ iteratively for increasing values of $j$.
It can be terminated in any iteration $j$ according to the runtime requirements of the analysis.
At that point, the algorithm has executed for at most $\Otilde(n^{\omega}\cdot j)$ time, and is guaranteed to have computed all points-to relationships as witnessed by $(4,j)$-bounded programs.
Hence, 
(i)~a timeout  does not waste analysis time, and
(ii)~the obtained results provide measurable completeness guarantees.

It is believed that $\omega=2+o(1)$, in which case  \cref{them:bounded_apa_subcubic} yields a quadratic bound.
Given such an improvement, a natural question is whether sub-quadratic algorithms are possible when we restrict our attention to witnesses that are poly-logarithmically bounded.
Clearly this is not possible for $\APAPA$, as the size of the output can be $\Theta(n^2)$,
but the question becomes interesting in the case of $\SPAPA$ that has output size $\Theta(1)$.
We answer this question in negative.

\smallskip
\begin{restatable}{theorem}{themovhardness}\label{them:ovhardness}
$(\All, \Otilde(1))$-bounded $\SPAPA$ has no sub-quadratic algorithm under the Orthogonal Vectors hypothesis.
\end{restatable}

Recall that in the $(\All, \Otilde(1))$-bounded version of the problem, we restrict our attention to witness programs of poly-logarithmic length (i.e., on all statement types).
The above bound holds even when $m=O(n)$.
Orthogonal Vectors is a well-studied problem with a long-standing quadratic worst-case upper bound.
The corresponding hypothesis states that there is no sub-quadratic algorithm for the problem~\cite{Williams19}.
It is also known that the strong exponential time hypothesis (SETH) implies the Orthogonal Vectors hypothesis~\cite{Williams05}.
Under \cref{them:ovhardness}, our algorithm from \cref{them:bounded_apa_subcubic} is \emph{optimal}  when $\omega=2+o(1)$.

Finally, to establish \cref{them:bounded_apa_subcubic}, we solve $\DReachability{1}$ in nearly matrix-multiplication time.

\smallskip
\begin{restatable}{theorem}{themd1}\label{them:d1}
All-pairs $\DReachability{1}$ is solvable in $\Otilde(n^{\omega})$ time, where $n$ is the number of nodes and $\omega$ is the matrix-multiplication exponent.
\end{restatable}

Observe the different regimes of $\DReachability{k}$ for various values of $k$.
When $k=0$, the problem becomes standard graph reachability, which is solvable in $O(n^{\omega})$ time~\cite{Munro71}.
When $k\geq 2$, the problem is solvable in $O(n^3)$ time, and this bound is believed to be tight (wrt polynomial improvements)~\cite{Heintze97}.
The case of $k=1$ was recently solved independently in~\cite{Bradford2018}.
However, our algorithm behind~\cref{them:d1} is more straightforward:~it establishes a purely combinatorial reduction of the problem to $O(\log^2 n)$ many transitive-closure operations.
From there, it relies on algebraic, fast-matrix multiplication for performing each transitive closure in $O(n^{\omega})$ time.
In contrast, the algorithm in~\cite{Bradford2018} is considerably longer and relies on intricate algebraic transformations.
In addition, our algorithm is a $\log n$-factor faster.
We refer to \cref{subsec:d1_subcubic} for a more detailed comparison.

\Paragraph{Parallelizability of $\APA$.}
Parallelization is an important aspect of the complexity of a problem.
In the case of $\APA$, this is further motivated by a multitude of parallel implementations~\cite{MendezLojo2010,Mendez12,Su14,Wang17,Liu19,Blass19}.
The question is thus whether $\APA$ is effectively parallelizable.
We answer this question in negative.
\smallskip
\begin{restatable}{theorem}{thempcomplete}\label{them:pcomplete}
$\SPAPA$ is P-complete.
\end{restatable}

\cref{them:pcomplete} implies that all efforts on complete parallelization must stay at the current, heuristic level, while their  improvements vanish on hard instances.
On the other hand, in spirit similar to \cref{them:bounded_apa_subcubic}, we can seek for mild restrictions to the problem that make it parallelizable.
Recall that, given some $i\in \Nats$, the complexity class NC$^i$ contains the problems that are solvable in parallel time $O(\log^i n)$ with polynomially many processors.
We show the following theorem.

\smallskip
\begin{restatable}{theorem}{themboundedparallelizable}\label{them:bounded_parallelizable}
$(4,\log^i n)$-bounded $\APAPA$ is in NC$^{i+2}$.
\end{restatable}
Thus, $\APA$ with poly-logarithmically many applications of type~4 statements is highly parallelizable. 
Together, \cref{them:pcomplete} and \cref{them:bounded_parallelizable} expose the core hardness in the parallelization of $\APA$ as stemming from handling statements of type~4.


In the following sections we present details of the above theorems.
To improve readability, in the main paper we present algorithms, examples, and proofs of all theorems.
To highlight the main steps of the proofs, we also present all intermediate lemmas;  many lemma proofs, however, are relegated to the appendix.

\section{A Baseline Algorithm For Andersen's Pointer Analysis}\label{sec:cubic}

In this section, we present a baseline algorithm for $\APA$.
It has similar flavor as other algorithms in the literature, but allows us to more easily establish \cref{them:cubic_upper_bound}.



\smallskip
\begin{algorithm}[H]
\small
\DontPrintSemicolon
\SetInd{0.3em}{0.3em}
\caption{$\AlgoCubic$}\label{algo:cubic}
\SetKwFunction{ProcessEps}{$\ProcessEpsilon$}
\SetKwProg{Fn}{Function}{:}{}
\KwIn{
An instance $(A,S)$ of $\APA$.
}
\KwOut{A set $(a,b)$ of all points-to relationships $b\in \PointsToSet{a}$.}
\BlankLine
\begin{multicols}{3}
\tcp{Initialization}
Let $\Worklist, \Done\gets \emptyset$\\
\ForEach{$a=\&b$}{\label{line_cubic_init1}
Insert $(b,a,\&)$ in $\Worklist$\\
Insert $(b,a,\&)$ in $\Done$
}
\ForEach{$a=b$}{\label{line_cubic_init2}
Insert $(b,a,\epsilon)$ in $\Worklist$\\
Insert $(b,a,\epsilon)$ in $\Done$
}
\tcp{Computation}
\While{$\Worklist\neq \emptyset$}{\label{line:cubic_main_loop}
Extract $(a,b,t)$ from $\Worklist$\\
\eIf{$t=\epsilon$}{
$\ProcessEpsilon(a,b)$
}
{
$\ProcessRef(a,b)$
}
}
\Return{$\{(a,b)\colon (b,a,\&) \in \Done \}$}
\pagebreak

\tcp{Process inclusion edges}
\Fn{\ProcessEps{$a$, $b$}}{
\ForEach{$(c,a,\&)\in\Done$}{\label{line:proceps3}
$\Establish(c,b,\&)$
}
}

\BlankLine
\BlankLine
\BlankLine
\tcp{Establish new inclusion and reference edges}
\SetKwFunction{ProcessEstablish}{$\Establish$}
\Fn{\ProcessEstablish{$a,b,t$}}{
\uIf{$(a,b,t)\not \in\Done$}{
Insert $(a,b,t)$ in $\Worklist$\\
Insert $(a,b,t)$ in $\Done$\\
}
}

\pagebreak
\tcp{Process reference edges}
\SetKwFunction{ProcessEps}{$\ProcessRef$}
\Fn{\ProcessEps{$a$, $b$}}{
\ForEach{$(b,c,\epsilon)\in\Done$}{\label{line:procref1}
$\Establish(a, c, \&)$
}
\ForEach{$c=*b$ in $S$}{\label{line:procref2}
$\Establish(a, c, \epsilon)$
}
\ForEach{$*b=c$ in $S$}{\label{line:procref3}
$\Establish(c, a, \epsilon)$
}
}
\end{multicols}
\end{algorithm}

\Paragraph{Algorithm $\boldsymbol{{\AlgoCubic}}$.}
The algorithm performs a form of dynamic transitive closure of a Dyck graph, in similar spirit to existing algorithms in the literature.
The key difference is that instead of just maintaining inclusion relationships $a\DEdge{\epsilon}b$, the flow graph also explicitly captures points-to relationships $a\DEdge{\&}b$ in its edges.
In each iteration, the algorithm processes a newly inserted edge $a\DEdge{t}b$, where $t\in \{\epsilon, \& \}$, and inserts new edges that represent flows-into and inclusion relationships that are implied by $a\DEdge{t}b$, the current state of the flow graph, and the statements in $S$.
The algorithm can be seen as an on-the-fly version of the difference-propagation technique~\cite{Pearce04,Sridharan09}, without the need to store difference-sets explicitly.
See \cref{algo:cubic} for a detailed description.

\begin{proof}[Proof of \cref{them:cubic_upper_bound}]
The correctness of the algorithm follows by a straightforward induction.
Here we argue about the complexity.
The initialization clearly runs in time $O(m)=O(n^2)$.
For every pair of pointers $(a,b)$, the worklist $\Worklist$ can have at most one element of the form $(a,b,\epsilon)$ and at most one element of the form $(a,b,\&)$, due to the set $\Done$.
Hence, the main loop in \cref{line:cubic_main_loop} is executed at most twice for every pair of pointers $(a,b)$, and thus $O(n^2)$ times in total.
For every such pair, the loops in \cref{line:proceps3,line:procref1} are executed once for each pointer $c$.
Hence, these loops are executed in $O(n^3)$ time in total.
Finally, each of the loops in \cref{line:procref2,line:procref3} is executed once for every pointer $a$.
Summing over all statements of type~3 and type~4, which are $O(m)$ many, we obtain that these loops are executed $O(n\cdot m)$ times in total.
Since $m=O(n^2)$, we obtain the desired bound $O(n^3)$.
\end{proof}

\section{The Cubic Hardness of Andersen's Pointer Analysis}\label{sec:triangle_hard}

In this section we prove \cref{them:bmm_hard}, i.e., that even $\SPAPA$ (given pointers $a,b$, is it that $b\in \PointsToSet{a}$?) does not have a combinatorial sub-cubic algorithm under the BMM hypothesis.

\Paragraph{Exhaustive vs On-Demand $\APA$.}
Recall the difference between $\APAPA$ and $\SPAPA$.
The former problem asks for the points-to set $\PointsToSet{a}$ of \emph{every} pointer $a$, while the latter focuses on a \emph{specific} pair of pointers $a,b$, and asks whether $b\in \PointsToSet{a}$.
Thus, $\SPAPA$ is a special version of $\APAPA$.
The (combinatorial) cubic hardness of $\APAPA$ follows straightforwardly via a reduction from the graph transitive closure, which has the same combinatorial cubic lower bound.
Indeed, given a directed graph $H=(V,E)$, we simply create an $\APA$ instance $(A,S)$, where $A=\{a_1, a_2\colon a\in V \}$ contains two copies of each node in $V$, and $S=S_1\cup S_2$, where $S_1=\{a_1=b_1\colon (b,a)\in E \}$ and $S_2=\{a_1=\& a_2\colon a\in V \}$.
That is, the type-1 statements directly correspond to the edges in $H$, and the type-2 statements are dummy statements that initialize the points-to sets. It follows immediately that after solving the $\APA$ instance $(A,S)$, for every node $a\in V$, the points-to set $\PointsToSet{a_1}$ contains all $b_2$ such that $a$ is reachable by $b$ in $H$.
Notably, the $\APA$ instance does not even make use of type-3 and type-4 statements.

We remark that the reduction does not apply for the on-demand version of $\APA$.
In the following we establish the proof of \cref{them:bmm_hard}.
In this direction, we establish a fine-grained reduction~\cite{Williams19} from the problem of deciding whether a graph contains a triangle to $\SPAPA$.

\Paragraph{Reduction from finding triangles.}
Consider an undirected graph $H=(V,E)$, of $n'$ nodes, where the task is to determine  if $H$ contains a triangle.
For notational convenience, we take the node set of $H$ to be the set of integers $[n']$.
Hence, the task is to determine if there exist distinct $i,j,k\in [n']$ such that $(i,j), (j,k), (k,i)\in E$.
Our reduction constructs four pointers $a_i, b_i, c_i, d_i$ for every node $i'\in [n']$, 
and uses one additional pointer $s$ such that $s\in \PointsToSet{c_1}$ iff $H$ has a triangle.

\SubParagraph{Intuition.}
The search for a triangle $(i,j,k)$ of $H$ can be seen as a search for two nodes $i$ and $k$ such that $k$ is both a distance-1 and distance-2 neighbor of $i$.
In our reduction, the two pointers $c_k$ and $d_k$ are such that $c_k\DEdge{\epsilon} d_k$ in the resolved Dyck graph $\ov{G}$ iff $k$ is both distance-1 and distance-2 neighbor of some node $i$.
This is achieved in two steps.
\begin{compactenum}
\item In the initial Dyck graph $G$ of the $\APA$ instance, $c_k$ flows into $a_i$, by introducing two statements $b_j=\& c_k$ and $a_i=b_j$, where $j$ is a neighbor of both $i$ and $k$.
\item We have a statement $*a_i=d_k$.
\end{compactenum}
We also introduce some additional statements between all $d_i$ and between all $c_i$ such that $c_1$ is D-reachable from $d_1$ in $\ov{G}$ iff there exists some $k$ such that $c_k\DEdge{\epsilon} d_k$ in $\ov{G}$
(i.e., by the above, $k$ is a distance-1 and distance-2 neighbor of some node $i$).
Finally, we have a statement $d_1=\& s$, so that $s\in \PointsToSet{c_1}$ iff the above condition holds.
\cref{fig:triangle_hard} provides an illustration.

\SubParagraph{Formal construction.}
We now proceed with the formal construction.
We construct an instance of $\SPAPA$ as follows.
\begin{compactenum}
\item We introduce a distinguished pointer $s$.
\item For every node $i\in [n']$, we introduce four pointers $a_i, b_i, c_i, d_i$.
\item For every $(i,j)\in E$ with $j<i$, we have
(i)~$a_i=b_j$, and
(ii)~$b_i=\& c_j$.
\item For every $(i,j)\in E$ with $j>i$, we have
$*a_i=d_j$.
\item Finally, we have the following sets of assignments.
\begin{align*}
d_1=\& s \quad d_2=\&d_1   \quad  \cdots \quad d_n=\&d_{n-1}  \qquad \text{and}  \qquad
c_1=*c_2 \quad c_2=*c_3\quad \cdots \quad c_{n-1}=*c_n\ .
\end{align*}
\end{compactenum}
The on-demand question is whether $s\in \PointsToSet{c_1}$.
Observe that the number of pointers of our $\APA$ instance is $O(n')$, and the above construction can be easily carried out in time proportional to the size of $H$.

\Paragraph{Correctness.}
We now establish the correctness of the above construction.
The key idea is as follows.
Recall our definition of the resolved Dyck graph $\ov{G}$ from \cref{subsec:dyck_reachability}.
An edge $d_k\DEdge{\epsilon} c_k$ is inserted in $\ov{G}$ iff $k$ is both a distance1- and distance-2 neighbor of some node $i$.
In turn, this implies the existence of a triangle in $H$ that contains $i$ and $k$.
We have the following lemma.
\begin{figure}
\newcommand{\xdisposition}{0}
\newcommand{\ydisposition}{0}
\newcommand{\xtstep}{0.9}
\newcommand{\ytstep}{0.5}
\newcommand{\xstep}{1.2}
\newcommand{\ystep}{1.5}
\newcommand{\gatex}{0.8}
\newcommand{\gatey}{0.4}

\newcommand{\drawGate}[3]{
\node[] (T#1) at (#2 + \gatex, #3 + \gatey/2)  {$T_{#1}$};
\node[] (#1bar) at (#2 + \gatex, #3 - \gatey/2)  {$\bar{#1}$};
\node[] (#1) at (#2,#3)  {$#1$};
}
\newcommand{\drawInput}[3]{
\node[] (T#1) at (#2, #3 )  {$T_{#1}$};
\node[] (#1) at (#2+ 1.25*\gatex,#3)  {$#1$};
}
\centering
\begin{tikzpicture}[thick,
pre/.style={<-,shorten >= 2pt, shorten <=2pt,  thick},
post/.style={->, thick},
seqtrace/.style={->, line width=2},
postpath/.style={->, thick, decorate, decoration={zigzag,amplitude=1pt,segment length=2mm,pre=lineto,pre length=2pt, post=lineto, post length=4pt}},
node/.style={very thick, draw=black, circle, inner sep = 2, minimum size=2mm},
]

\foreach \x in {1,...,5}{
\node[] at (\x*\xstep, 0*\ystep) (a\x) {$a_{\x}$};
\node[] at (\x*\xstep, -1*\ystep) (b\x) {$b_{\x}$};
\node[] at (\x*\xstep, -2*\ystep) (c\x) {$c_{\x}$};
\node[] at (\x*\xstep, -3*\ystep) (d\x) {$d_{\x}$};
}
\node[] at  (0*\xstep, -3*\ystep) (s) {$s$};

\node[] at(6.5*\xstep, 0*\ystep)	{
$\begin{aligned}
*a_3=&d_2\\
\end{aligned}$
};

\node[] at(6.5*\xstep, -1*\ystep)	{
$\begin{aligned}
*a_4=&d_2\\
*a_4=&d_3\\
\end{aligned}$
};

\node[] at(6.5*\xstep, -2*\ystep)	{
$\begin{aligned}
*a_5=&d_1\\
*a_5=&d_3
\end{aligned}$
};

\draw[post] (b2) to (a3);
\draw[post] (b2) to (a4);
\draw[post] (b3) to (a4);
\draw[post] (b3) to (a5);

\draw[post] (c2) to node[midway, above, sloped]{$\&$} (b3);
\draw[post] (c2) to node[midway, above, sloped]{$\&$} (b4);
\draw[post] (c3) to node[midway, above, sloped]{$\&$} (b4);
\draw[post] (c3) to node[midway, above, sloped]{$\&$} (b5);

\draw[post, smooth, looseness=0.5, in=-125] (c1) to node[midway, below, sloped]{$\&$} (b5);
\draw[post, smooth, looseness=0.5, in=-125] (b1)   to  (a5);

\draw[post] (s) to node[midway, above, sloped]{$\&$} (d1);
\draw[post] (d1) to node[midway, above, sloped]{$\&$} (d2);
\draw[post] (d2) to node[midway, above, sloped]{$\&$} (d3);
\draw[post] (d3) to node[midway, above, sloped]{$\&$} (d4);
\draw[post] (d4) to node[midway, above, sloped]{$\&$} (d5);

\draw[post] (c2) to node[midway, below, sloped]{$*$} (c1);
\draw[post] (c3) to node[midway, below, sloped]{$*$} (c2);
\draw[post] (c4) to node[midway, below, sloped]{$*$} (c3);
\draw[post] (c5) to node[midway, below, sloped]{$*$} (c4);

\begin{scope}[shift={(-4.4*\xstep,-1.45*\ystep)}]
\renewcommand{\xstep}{1.3}
\renewcommand{\ystep}{0.65}

\node[node] (1) at (0*\xstep, 0*\ystep) {$1$};
\node[node] (2) at (1*\xstep, 1*\ystep) {$2$};
\node[node] (3) at (2*\xstep, 1*\ystep) {$3$};
\node[node] (4) at (2*\xstep, -1*\ystep) {$4$};
\node[node] (5) at (1*\xstep, -1*\ystep) {$5$};

\draw[very thick] (1) to (5);
\draw[very thick] (2) to (3);
\draw[very thick] (2) to (4);
\draw[very thick] (3) to (4);
\draw[very thick] (3) to (5);

\end{scope}

\end{tikzpicture}
\caption{
An undirected graph (left) and the corresponding $\SPAPA$ instance (right).
}
\label{fig:triangle_hard}
\end{figure}
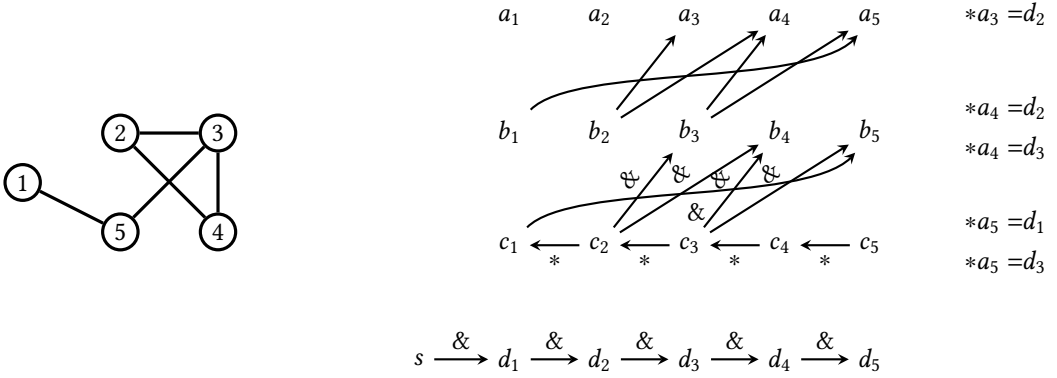

\begin{restatable}{lemma}{lemtrianglecorrectness}\label{lem:triangle_correctness}
We have that $s$ flows into $c_1$ in $\ov{G}$ iff $H$ has a triangle.
\end{restatable}
\begin{proof}
We prove each direction separately.

\noindent{\em $(\Rightarrow)$.}
Assume that $H$ has a triangle $(i,j,k)$, with $i>j>k$.
Then we have $a_i=b_j$ and $b_j=\&c_k$ and thus $c_k\DPath{\&}a_i$ in $\ov{G}$.
In addition, we have $*a_i=d_k$, and thus $d_k\DEdge{\epsilon} c_k$ in $\ov{G}$.
Observe that this creates a path $s\DPath{k \&} d_k \DPath{\epsilon} c_k \DPath{(k-1) *} c_1$, which witnesses that $s$ flows into $c_1$ in $\ov{G}$.

\noindent{\em $(\Leftarrow)$.}
Assume that $s$ flows into $c_1$.
Observe that for all $a_i$, if some node $x$ flows into $a_i$ then $x$ is a $c$ node.
It follows that $\ov{G}$ is identical to $G$ with some additional edges from $d$ nodes to $c$ nodes.
Hence, since $s$ flows into $c_1$, there exists some $k\in[n']$ such that $\ov{G}$ has an edge $d_k\DEdge{\epsilon}c_k$.
This means that (i)~there exists an $i$ such that $c_k$ flows into $a_i$ (thus $k$ is a distance-2 neighbor of $i$), and 
(ii)~there is a statement $*a_i=d_k$ (thus $k$ is a distance-1 neighbor of $i$).
Hence $H$ has a triangle containing $i$ and $k$.
The desired result follows.
\end{proof}

We conclude this section with the proof of \cref{them:bmm_hard}.

\begin{proof}[Proof of \cref{them:bmm_hard}.]
Due to \cref{lem:triangle_correctness}, we have that $s$ flows into $c_1$ iff $H$ contains a triangle.
By \cref{lem:resolved_graph} we have that $s\in \PointsToSet{c_1}$ iff $H$ contains a triangle.
By~\cite{Williams18b}, triangle detection has no sub-cubic combinatorial algorithm under the combinatorial BMM-hypothesis.

The desired result follows.
\end{proof}
\section{A Sub-cubic Algorithm for Bounded Andersen's Pointer Analysis}\label{sec:bounded}

In this section, we first show \cref{them:bounded_apa_subcubic}, i.e., that computing points-to relationships when bounding the number of applications of type~4 statements admits a sub-cubic algorithm.
To this end, we first prove in \cref{subsec:d1_subcubic} \cref{them:d1}, i.e., that $\DReachability{1}$ can be solved in nearly matrix-multiplication time.
Afterwards, we use this result to prove \cref{them:bounded_apa_subcubic} in \cref{subsec:bounded_subcubic}.

\subsection{A Sub-cubic Algorithm for $\DReachability{1}$}\label{subsec:d1_subcubic}

In this section we establish a combinatorial reduction of all-pairs $\DReachability{1}$ to $O(\log^2 n)$ matrix multiplications,
establishing that the problem is solvable in nearly matrix-multiplication time.
We first set up some helpful notation, and then present the main algorithm.
Consider a Dyck graph $G=(V,E)$.

\Paragraph{Path indexing.}
Consider a path $P=x_1,\dots, x_{l}$.
Given some $i\in [l]$, we denote by $P[i]=x_i$.
Given $i,j\in [l]$, with $i\leq j$, we denote by $P[i:j]=x_i,\dots, x_j$.
For simplicity, we let $P[:j]=P[1:j]$ and $P[i:]=P[i:l]$.

\Paragraph{Stack heights.}
Consider a path $P$ of length $l$.
We denote by $\NumOpen(P)$ (resp., $\NumClose(P)$) the number of $\&$ (resp., $*$) symbols that appear in the label $\Label(P)$.
The \emph{stack height} of $P$ is defined as $\StackHeight(P)=\NumOpen(P) - \NumClose(P)$ (note that we can have $\StackHeight(P)<0$).
The \emph{maximum stack height} of a path is defined as $\MaxStackHeight(P)=\max_{P'}\StackHeight(P')$, where $P'$ ranges over prefixes of $P$.

\Paragraph{Monotonicity and local maxima.}
Consider a path $P$ of length $l$.
We say that $P$ is \emph{monotonically increasing} (resp., \emph{monotonically decreasing}) if for all $i$ with $1\leq i<l$, we have $\StackHeight(P[:i])\leq \StackHeight(P[:i+1])$
(resp., $\StackHeight(P[:i])\geq \StackHeight(P[:i+1])$).
Given some $i$ with $1\leq i\leq l$, we say that $P$ has a \emph{local maxima} in $i$ if the following conditions hold.
\begin{compactenum}
\item Either $i=1$ or $\StackHeight(P[:i-1])< \StackHeight(P[:i])$.
\item For every $j>i$ such that $\StackHeight(P[:i]) < \StackHeight(P[:j])$, there exists some $l$ with $i<l<j$ such that $\StackHeight(P[:l])<\StackHeight(P[:i])$.
\end{compactenum}

\Paragraph{Bell-shape-reachability.}
We call a path $P$ \emph{bell-shaped} if it has exactly one local maxima.
If $P$ is bell-shaped, it can be decomposed as $P\colon P_1\circ P_2$ where
$P_1$ (resp., $P_2$) is a monotonically increasing (resp., monotonically decreasing) path.
Consider two nodes $x$, $y$.
We say that  $y$ is \emph{$i\&$-reachable} (resp., $i*$-reachable) from $x$, for some $i\in \Nats$,
if there is a path $x\DPath{i\&} y$ (resp., $x\DPath{i*} y$).
We say that $y$ is \emph{bell-shape-reachable} from $x$ if there exists a bell-shaped path $x\DPath{} y$.

\Paragraph{Node distances.}
Given two nodes $x,y$, we define the \emph{distance} $\Distance(x,y)$ from $x$ to $y$ as 
the length of the shortest path $P\colon x\DPath{}y$ with $\Label(P)\in \Dyck_1$, if such a path exists, otherwise $\Distance(x,y)=\infty$.
The \emph{maxima-distance} $\MaximaDistance(x,y)$ is the smallest number of local maxima among all shortest paths $x\DPath{}y$.
The following known lemma states that the distances between two reachable nodes is at most quadratic.
Note that $\MaximaDistance(x,y)\leq \Distance(x,y)$, hence the same bound holds for the maxima-distance.

\smallskip
\begin{lemma}[\cite{Deleage86}]\label{lem:d1_distance}
For every $x,y\in V$, if $\Distance(x,y)<\infty$ then $\Distance(x,y)=O(n^2)$.
\end{lemma}

\Paragraph{Routine $\BellReachAlgo$.}
The main component of our algorithm for $\DReachability{1}$ is a routine $\BellReachAlgo$ that computes bell-shape-reachability.
Given an input Dyck graph $G=(V,E)$, $\BellReachAlgo$ computes all pairs of nodes $(x,y)$ such that $y$ is bell-shape-reachable from $x$.
The algorithm constructs a sequence of $O(\log L)$ plain (i.e., not Dyck) digraphs $(G_i=(K, R_i ))_i$, where $L$ is an upper bound on the distance $\Distance(x,y)$ of every pair of nodes $x,y\in V$, given by \cref{lem:d1_distance}.
The node set $K$ is common to all $G_i$ and consists of three copies $x_1, x_2, x_3$ for every node $x\in V$.

Intuitively, the algorithm performs a form of successive doubling on the length of the bell-shaped paths that witness reachability.
In iteration $i$, the algorithm performs all-pairs reachability in $G_i$, and using this reachability information, constructs the edge set $R_{i+1}$.
In high level, $G_i$ consists of three copies of the graph $G$, where bell-shaped paths of maximum stack height at most $2^i-1$ are summarized as $\epsilon$-labeled edges in the first and second copy.
Paths between the nodes in the first and third copy are used to summarize monotonically increasing and (resp., decreasing) paths in $G$ with labels of the form $2^{i} \&$ (resp., $2^{i}*)$.
We refer to \cref{algo:bell_reach} for a detailed description and to \cref{fig:phaseA} for an illustration.

\smallskip
\begin{algorithm}
\small
\DontPrintSemicolon
\SetInd{0.3em}{0.3em}
\caption{$\BellReachAlgo$}\label{algo:bell_reach}
\SetKwFunction{ProcessEps}{$\ProcessEpsilon$}
\SetKwProg{Fn}{Function}{:}{}
\KwIn{
A Dyck graph $G=(V,E)$
}
\KwOut{A set $\{(x,y)\}_{x,y\in V}$ such that $y$ is bell-shape-reachable from $x$.}
\BlankLine
\begin{multicols}{2}
\tcp{Initialization}
Construct a node set $K=\{ x_1, x_2, x_3\colon x\in V \}$\\
Construct an edge set $R_1$, initially $R_1\gets \emptyset$\\
\ForEach{$j\in[2]$}{
Insert $(x_j,y_j)\in R_1$ iff  $(x,y,\epsilon)\in E$\\
Insert $(x_j,y_{j+1})\in R_1$ iff $(x,y,\&)\in E$\\
Insert $(y_{j+1}, x_j)\in R_1$ iff $(x,y,*)\in E$\\
}
Construct the graph $G_1=(K, R_1)$\\
Let $L\gets$ an upper bound on $\Distance(x,y)$ for all $x,y\in V$\\

\pagebreak
\tcp{Computation}
\ForEach{$i\in [\ceil{\log L}]$}{
Compute all-pairs reachability in $G_i$\\
Construct an edge set $R_{i+1}$, initially $R_{i+1}\gets\emptyset$\\
\ForEach{$j\in[2]$}{
Insert $(x_j,y_i)\in R_{i+1}$ iff  $x_1\Path y_1$ in $G_i$\\
Insert $(x_j,y_{j+1})\in R_{i+1}$ iff $x_1\Path y_3$ in $G_i$\\
Insert $(y_{j+1}, x_j)\in R_{i+1}$ iff $y_3\Path x_1$ in $G_i$\\
}
Construct the graph $G_{i+1}=(K, R_{i+1})$\\
}
\Return{$R_{\ceil{\log L} +1}$}
\end{multicols}
\end{algorithm}

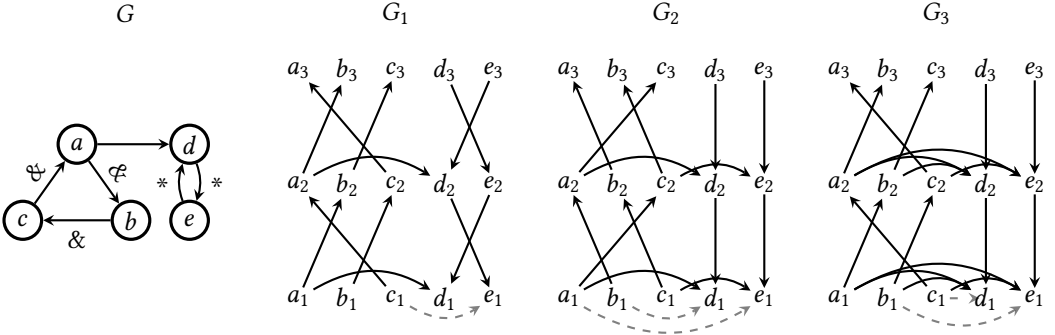
\begin{figure}
\newcommand{\xdisposition}{3.6}
\newcommand{\ydisposition}{0}
\newcommand{\xtstep}{0.9}
\newcommand{\ytstep}{0.5}
\newcommand{\xstep}{0.65}
\newcommand{\ystep}{1.5}
\newcommand{\gatex}{0.8}
\newcommand{\gatey}{0.4}

\centering
\begin{tikzpicture}[thick,
pre/.style={<-, thick, shorten >=-0.1cm,shorten <=-0.1cm},
post/.style={->, thick, shorten >=-0.1cm,shorten <=-0.1cm},
seqtrace/.style={->, line width=2},
postpath/.style={->, thick, decorate, decoration={zigzag,amplitude=1pt,segment length=2mm,pre=lineto,pre length=2pt, post=lineto, post length=4pt}},
node/.style={very thick, draw=black, circle, inner sep = 0, minimum size=5mm},
]

\node[] at (3*\xstep, 2.5*\ystep) {$G$};

\newcommand{\namesarray}{a, b, c, d, e}

\foreach \z in {1,2,3}{
\begin{scope}[shift={(\z*\xdisposition,0)}]
\pgfmathtruncatemacro\G{\z}
\node[] at (3*\xstep, 2.5*\ystep) {$G_{\G}$};
\foreach \x in {1,...,5}{
\node[] at (\x*\xstep, 0*\ystep) (\z0\x) {$\names[\x]_{1}$};
\node[] at (\x*\xstep, 1*\ystep) (\z1\x) {$\names[\x]_{2}$};
\node[] at (\x*\xstep, 2*\ystep) (\z2\x) {$\names[\x]_{3}$};
}
\end{scope}
}

\foreach \x in {0,1}{
\pgfmathtruncatemacro\xinc{\x+1}
\draw[post] (1\x1) to (1\xinc2);
\draw[post] (1\x2) to (1\xinc3);
\draw[post] (1\x3) to (1\xinc1);
\draw[pre] (1\x5) to (1\xinc4);
\draw[pre] (1\x4) to (1\xinc5);
\draw[post, bend left=30] (1\x1) to (1\x4);
\draw[post, bend left=30] (2\x1) to (2\x4);
\draw[post, bend left=30] (3\x1) to (3\x4);
\draw[post, bend left=30] (2\x3) to (2\x5);
\draw[post, bend left=30] (3\x3) to (3\x5);

\draw[post] (2\x1) to (2\xinc3);
\draw[post] (2\x2) to (2\xinc1);
\draw[post] (2\x3) to (2\xinc2);
\draw[pre] (2\x5) to (2\xinc5);
\draw[pre] (2\x4) to (2\xinc4);

\draw[post] (3\x1) to (3\xinc2);
\draw[post] (3\x2) to (3\xinc3);
\draw[post] (3\x3) to (3\xinc1);
\draw[pre] (3\x5) to (3\xinc5);
\draw[pre] (3\x4) to (3\xinc4);

\draw[post, bend left=30,] (3\x1) to (3\x5);
\draw[post, bend left=30] (3\x2) to (3\x4);

} 

\draw[post, bend right=30, opacity=0.5, dashed] (103) to (105);

\draw[post, bend right=30, opacity=0.5, dashed] (201) to (205);
\draw[post, bend right=30, opacity=0.5, dashed] (202) to (204);

\draw[post, bend right=30, opacity=0.5, dashed] (302) to (305);
\draw[post, bend right=0, opacity=0.5, dashed] (303) to (304);

\begin{scope}[shift={(2*\xstep,1.35*\ystep)}]
\renewcommand{\xstep}{1}
\renewcommand{\ystep}{1}

\node[node] (a) at (0*\xstep, 0*\ystep){$a$};
\node[node] (b) at (+1.44/2*\xstep, -1*\ystep){$b$};
\node[node] (c) at (-1.44/2*\xstep, -1*\ystep){$c$};
\node[node] (d) at (1.5*\xstep, 0*\ystep){$d$};
\node[node] (e) at (1.5*\xstep, -1*\ystep){$e$};

\draw[post, shorten >=-0cm,shorten <=-0cm] (a) to node[above, midway, sloped]{$\&$} (b);
\draw[post, shorten >=-0cm,shorten <=-0cm] (b) to node[below, midway, sloped]{$\&$} (c);
\draw[post, shorten >=-0cm,shorten <=-0cm] (c) to node[above, midway, sloped]{$\&$} (a);

\draw[post, shorten >=-0cm,shorten <=-0cm] (a) to node[above, midway, sloped]{} (d);

\draw[post, bend left=20, shorten >=-0cm,shorten <=-0cm] (d) to node[above, midway, sloped]{$*$} (e);
\draw[post, bend left=20, shorten >=-0cm,shorten <=-0cm] (e) to node[above, midway, sloped]{$*$} (d);

\end{scope}

\end{tikzpicture}
\caption{
Illustration of $\BellReachAlgo$ on the Dyck graph $G$ (left).
Bell-shape-reachability in $G$ as witnessed by paths $P\colon x\DPath{}y$ with $\MaxStackHeight(P)\leq 2^{i}-1$ is captured in graph $G_i$ (right) by the path $x_1\Path y_1$.
Dashed edges in $G_i$ represent the summarization of the path, which is carried over to $G_{i+1}$ as a single edge.
}
\label{fig:phaseA}
\end{figure}

\SubParagraph{Correctness of $\BellReachAlgo$.}
It is straightforward that for each iteration $i$, if $x_1\Path y_1$ in $G_i$, then $y$ is D-reachable from $x$ in $G$.
The following lemma captures the inverse direction restricted to bell-shaped paths,
i.e., if $x\DPath{\epsilon}y$ via a bell-shaped path $P$ in $G$ with $\MaxStackHeight(P)\leq 2^i-1$, then $x_1\Path y_1$  in $G_i$.
The key invariants are stated in the following lemma.

\smallskip
\begin{restatable}{lemma}{lembellshapecorrectness}\label{lem:bell_shape_induction}
Consider an execution of the routine $\BellReachAlgo$.
For each $i\in [\ceil{\log L}]$, the following assertions hold.
\begin{compactenum}
\item\label{item:ind_epsilon}  If $y$ is D-reachable from $x$ via a bell-shaped path $P$ in $G$ with $\MaxStackHeight(P)\leq 2^i-1$, then $x_1\Path y_1$  in $G_i$.
\item\label{item:ind_increase} If $x\DPath{2^{i}\&} y$ via a monotonically increasing path in $G$ where the last edge is $\&$-labeled, then $x_1\Path y_3$ in $G_i$.
\item\label{item:ind_decrease} If $y\DPath{2^{i}*} x$ via a monotonically decreasing path in $G$ where the first edge is $*$-labeled, then $y_3\Path x_1$ in $G_i$.
\end{compactenum}
\end{restatable}


\Paragraph{Algorithm $\DOneAlgo$.}
We are now ready to describe our algorithm $\DOneAlgo$ for $\DReachability{1}$.
The algorithm performs $\ceil{\log L}$ iterations of $\BellReachAlgo$, where $L$ is an upper bound on the distance between any two reachable nodes in $G$.
See \cref{algo:done} for a detailed description.

\smallskip
\begin{algorithm}
\small
\DontPrintSemicolon
\SetInd{0.3em}{0.3em}
\caption{$\DOneAlgo$}\label{algo:done}
\SetKwProg{Fn}{Function}{:}{}
\KwIn{
A Dyck graph $G=(V,E)$
}
\KwOut{A set $\{(x,y)\}_{x,y\in V}$ such that $y$ is D-reachable from $x$.}
\BlankLine
\begin{multicols}{2}
\tcp{Initialization}
Let $G_1=(V, E_1)$ be a Dyck graph with $E_1=E$\\

\pagebreak
\tcp{Computation}
\ForEach{$i\in [\ceil{\log L} + 1]$}{
Let $X=\BellReachAlgo$ on input $G_i$\\
Let $E_{i+1}=E_{i} \cup \{ (x,y,\epsilon)\colon (x,y)\in X \}$\\
Construct the graph $G_{i+1}=(V, E_{i+1})$
}
\Return{$E_{\ceil{\log L} +1}$}
\end{multicols}
\end{algorithm}

\Paragraph{Correctness of $\DOneAlgo$.}
We now establish the correctness of $\DOneAlgo$.
We start with an intuitive description of the correctness, and afterwards we make the argument formal (see \cref{fig:path_summarization} for an illustration).

\begin{figure}
\def \scale{1}
\newcommand{\xstep}{1}
\newcommand{\ystep}{1.2}
\newcommand{\ybias}{-0.3}
\centering
\begin{tikzpicture}
\begin{axis}[scale=\scale,
axis line style = ultra thick,
compat=newest,
unit vector ratio*=1 1.2 1,
height=4.5cm,
width=\textwidth,
axis lines=left,
ymax=4.5,
ymin=0,
xmax=28,
xmin=0,
grid style={line width=0.5pt, draw=gray!50},
axis line style={-latex},
extra x ticks={3, 7, 10, 11, 15, 18, 22, 25},
extra x tick labels={$j_1$,{$l_1$, $j_2$},$l_2$, $j_3$, $l_3$, $j_4$, {$l_4$, $j_5$}, $l_5$},
xtick=\empty,
ytick=\empty,
yticklabels={},
xticklabels={},
xlabel={ Path Index},
ylabel={ Stack Height}
]

\addplot[color=black,mark=.,line width=1.5] coordinates {
(0,0)
(1,1)
(2,2)
(3,3)
(4,4)
(5,4)
(6,3)
(7,3)
(8,4)
(9,4)
(10,3)
(11,2)
(12,2)
(13,3)
(14,3)
(15,2)
(16,1)
(17,1)
(18,2)
(19,3)
(20,3)
(21,2)
(22,2)
(23,3)
(24,3)
(25,2)
(26,1)
(27,0)
};
\addplot[shift={(0.0,\ybias)}, color=gray,mark=.,line width=1.5] coordinates {
(0,0)
(1,1)
(2,2)
(3,3)
};
\addplot[shift={(0.0,\ybias)}, color=gray,mark=.,line width=1.5] coordinates {
(10,3)
(11,2)
};
\addplot[shift={(0.0,\ybias)}, color=gray,mark=.,line width=1.5] coordinates {
(15,2)
(16,1)
(17,1)
(18,2)
};
\addplot[shift={(0.0,\ybias)}, color=gray,mark=.,line width=1.5] coordinates {
(25,2)
(26,1)
(27,0)
};
\addplot[-latex,shift={(0.0,\ybias)}, , color=gray,dashed,line width=1.5] coordinates {
(3,3)
(7,3)
};
\addplot[-latex,shift={(0.0,\ybias)}, , color=gray,dashed,line width=1.5] coordinates {
(7,3)
(10,3)
};
\addplot[-latex,shift={(0.0,\ybias)}, , color=gray,dashed,line width=1.5] coordinates {
(11,2)
(15,2)
};
\addplot[-latex,shift={(0.0,\ybias)}, , color=gray,dashed,line width=1.5] coordinates {
(18,2)
(22,2)
};
\addplot[-latex,shift={(0.0,\ybias)}, , color=gray,dashed,line width=1.5] coordinates {
(22,2)
(25,2)
};
\end{axis}
\end{tikzpicture}
\caption{
Illustration of a path $P$ in graph $G_i$ (black) and its summarization path $P'$ in graph $G_{i+1}$ (gray).
The number of local maxima in $P'$ is at most half of that in $P$.
}\label{fig:path_summarization}
\end{figure}
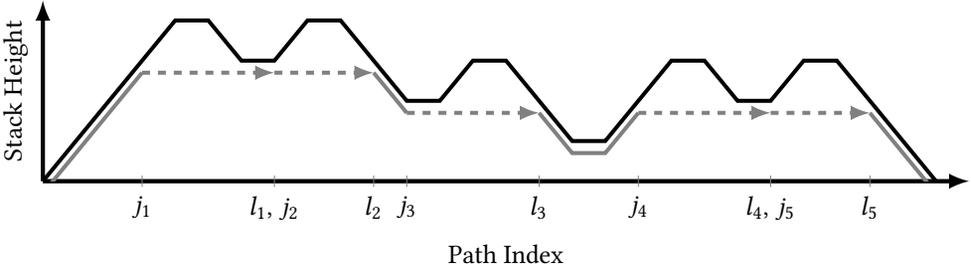

\SubParagraph{Intuitive argument of correctness.}
Consider an iteration $i$, and let $x,y\in V$ such that $y$ is D-reachable from $x$ via a path $P\colon x\Path y$.
At the end of the iteration, due to the execution of routine $\BellReachAlgo$, all bell-shaped paths $u\DPath{}v$ in $G_i$ are summarized as $\epsilon$-labeled edges $u\DEdge{\epsilon} v$ in $G_{i+1}$.
Hence, $P$ is summarized by a path $P'$ in $G_{i+1}$, where the bell-shaped sub-paths of $P$ are replaced by $\epsilon$-edges in $P'$.
How many times do we need to perform this iteration until the whole of $P$ is summarized by a single $\epsilon$-labeled edge?
The key insight is that the number of local maxima in $P'$ is \emph{at most half} of that in $P$.
Hence, it suffices to compute bell-shape-reachability a number of times that is logarithmic in the maxima-distance $\MaximaDistance(x,y)$.
Since $\MaximaDistance(x,y)\leq \Distance(x,y)$ and $\Distance(x,y)\leq L=O(n^2)$ (by \cref{lem:d1_distance}), $\ceil{\log L}=O(\log n)$ iterations suffice.

\SubParagraph{Formal correctness.}
We now proceed to make the above argument formal.
Given some iteration $i$ of $\DOneAlgo$, consider the graphs $G_i$ and $G_{i+1}$.
Let $P\colon x\DPath{}y$ be any path that witnesses D-reachability of $y$ from $x$ in $G_i$.
Let $(j_{\ell}, l_{\ell})_{\ell}$ be the index pairs that mark bell-shaped sub-paths in $P$.
We require that each $(j_{\ell}, l_{\ell})_{\ell}$ is maximal, in the following way.
\begin{compactenum}
\item None of $P[j_{1}-1, l_{1}]$, $P[j_{1}, l_{1}+1]$ and $P[j_{1}-1, l_{1}+1]$ is bell-shaped.
\item If $\ell>1$,  then $P[j_{\ell}, l_{\ell}+1]$ is not bell-shaped, and if $P[j_{\ell}-1, l_{\ell}]$ or $P[j_{\ell}-1, l_{\ell}+1]$ is bell-shaped, then $j_{\ell}-1\leq l_{\ell-1}$.
\end{compactenum}
Intuitively, the first bell-shaped sub-path of $P$ is as long as possible, and every following bell-shaped sub-path is as long as possible provided that it does not overlap with the previous bell-shaped sub-path.
We decompose $P$ as 
\begin{align*}
P=P_{1}^{\downarrow}\circ P_{1}^{\uparrow} \circ P[j_{1}:l_1] \circ P_{2}^{\downarrow}\circ P_{2}^{\uparrow} \circ P[j_{2}:l_2] \circ \cdots \circ P[j_{k}, l_{k}] \circ P_{k+1}^{\downarrow}\ ,
\end{align*}
where each $P_{\ell}^{\downarrow}$ (resp., $P_{\ell}^{\uparrow}$) is a monotonically decreasing (resp., monotonically increasing) path.
Note that $P_{1}^{\downarrow}=\epsilon$.
Observe that $P$ has $k$ local maxima, one in each bell-shaped sub-path $P[j_{\ell}:l_{\ell}]$.
In $G_{i+1}$, the path $P$ is summarized by a path $P'$ identical to $P$, but with all the bell-shaped sub-paths $P[j_{\ell}:l_{\ell}]$ replaced by edges $x_{j_{\ell}} \xrightarrow{\epsilon}y_{j_{\ell}}$ (see \cref{fig:path_summarization} for an illustration).
Given some index $1\leq h \leq |P|$ with $h \not \in [j_{\ell}+1, l_{\ell}-1]$ for each $\ell \in [k]$, we denote by $f(h)$ the corresponding index in $P'$.

\smallskip
\begin{remark}\label{rem:same_stack_height}
For every index $h$ of $P'$, we have $\StackHeight(P'[:h])=\StackHeight(P[:f^{-1}(h)])$.
\end{remark}

We first have two technical lemmas.
The first lemma states that all local maxima in $P'$ appear on the first node of bell-shaped sub-paths of $P$.

\smallskip
\begin{restatable}{lemma}{lempmaxima}\label{lem:p_maxima}
Assume that $P'$ has a local maxima at some $h$.
Then $f^{-1}(h)=j_{\ell}$ for some $\ell\in [k]$.
\end{restatable}

The following lemma formalizes the following observation:~if the beginning of a bell-shaped sub-path of $P$ marks a local maxima for $P'$, then the beginning of the next bell-shaped sub-path of $P$ cannot mark a local maxima for $P'$.
This is shown by arguing that the two bell-shaped sub-paths of $P$ are next to each other, i.e., there are no monotonically decreasing and increasing paths separating them.

\smallskip
\begin{restatable}{lemma}{lemtwomaxima}\label{lem:two_maxima}
Assume that $P'$ has a local maxima at some $h$.
Then $P_{j_{\ell}+1}^{\downarrow}= P_{j_{\ell}+1}^{\uparrow}=\epsilon$, where $j_\ell=f^{-1}(h)$.
\end{restatable}

With \cref{lem:p_maxima} and \cref{lem:two_maxima}, we can now formalize the insight that the maxima-distance between any two nodes halves in each iteration of $\DOneAlgo$.
Given some iteration $i$ of the algorithm, we denote by $\MaximaDistance_{i}(x,y)$ the maxima-distance from $x$ to $y$ in the graph $G_i$.
We have the following lemma.

\smallskip
\begin{restatable}{lemma}{lemmaxdistancehalved}\label{lem:max_distance_halved}
For each $i\in [\ceil{\log L}]$, for any two nodes $x,y\in V$ such that $y$ is reachable from $x$ in $G$,
we have that $\MaximaDistance_{i+1}(x,y)\leq \MaximaDistance_{i}(x,y)/2$.
\end{restatable}

Finally, we prove \cref{them:d1}, i.e., that all-pairs $\DReachability{1}$ is solvable in $\Otilde(n^{\omega})$ time.

\begin{proof}[Proof of \cref{them:d1}]
We first argue about the correctness of $\DOneAlgo$.
It follows immediately from the correctness of the routine $\BellReachAlgo$ that if $\DOneAlgo$ returns that $y$ is D-reachable from $x$ then there is a path $y$ is D-reachable from $x$ in $G$.
Here we focus on the inverse direction, i.e., assume that there is a path $P\colon x\DPath{}y$ in $G$ with $\Label(P)\in \Dyck_1$,  and we argue that $\DOneAlgo$ returns that $y$ is D-reachable from $x$. 
Recall that $ \MaximaDistance_i(x,y)$ is the maxima-distance from $x$ to $y$ in the graph $G_i$ constructed by the algorithm in the $i$-th iteration.
We have 
\begin{align*}
 \MaximaDistance_1(x,y)=\MaximaDistance(x,y)\leq \Distance(x,y) \leq L\ ,
\end{align*}
where the last inequality follows from our choice of $L$ as an upper-bound on $\Distance(x,y)$.
By \cref{lem:max_distance_halved}, we have $\MaximaDistance_{i+1}(x,y)\leq \MaximaDistance_{i}(x,y)/2$ for each $i\in [\ceil{\log L}]$,
hence after $i=\ceil{\log L}$ iterations, we have $\MaximaDistance_{i}(x,y)=1$.
Thus, in the last iteration of the algorithm, $y$ is bell-shape-reachable from $x$, and by the correctness of the routine $\BellReachAlgo$  (\cref{lem:bell_shape_induction}),
$\BellReachAlgo$ will return that $y$ is D-reachable from $x$.
Thus $\DOneAlgo$ will return that $y$ is D-reachable from $x$ in $G$, as desired.

We now turn our attention to the complexity of $\BellReachAlgo$.
The algorithm performs $O(\log L)$ invocations to the routine $\BellReachAlgo$.
In each invocation, $\BellReachAlgo$ performs $O(\log L)=O(\log n)$ transitive closure operations on graphs with $O(n)$ nodes.
Using fast BMM~\cite{Munro71}, each transitive closure takes $O(n^{\omega})$ time. 
The total running time of $\DOneAlgo$ is $O(n^{\omega}\cdot \log^2L)=\Otilde(n^{\omega})$, as by \cref{lem:d1_distance}, we have $L=O(n^2)$.

The desired result follows.
\end{proof}

\Paragraph{A comparison note with~\cite{Bradford2018}.}
A sub-cubic bound for $\DReachability{1}$ was recently established independently in~\cite{Bradford2018}.
The crux of that algorithm is an elegant algebraic matrix encoding of \emph{flat} $\DReachability{1}$ to $O(\log n)$ AGMY matrix multiplications, each performed in $O(n^{\omega}\cdot \log n)$ time~\cite{Alon1997}.
Intuitively, flat $\DReachability{1}$ concerns reachability witnessed by sequentially composing bell-shaped paths, and the above reduction gives a $O(n^{\omega}\cdot \log^2 n)$ bound for the problem.
A second step solves flat $\DReachability{1}$ for $O(\log n)$ iterations, with some special treatment needed in each iteration for ensuring correct AGMY representation.
Composing the two steps yields a $O(n^{\omega}\cdot \log^3 n)$ bound for $\DReachability{1}$~\cite[Theorem~2]{Bradford2018}.
In comparison, the algorithm presented here is a $\log n$-factor faster, and relies on a purely combinatorial reduction to $O(\log n^2)$ BMMs, which are then performed in $O(n^{\omega})$ time using algebraic techniques.

\subsection{Bounded Andersen's Pointer Analysis in Sub-cubic Time}\label{subsec:bounded_subcubic}

In the previous section we saw that $\DReachability{1}$ can be solved in nearly BMM time, i.e., $\Otilde(n^{\omega})$.
In this section we show how we can use this result to speed-up bounded $\APAPA$, towards \cref{them:bounded_apa_subcubic}.
Recall that, given some $i\in [4]$ and $j\in \Nats$, the $(i,j)$-bounded $\APA$ asks to compute all memory locations $b$ that a pointer $a$ may point to,
as witnessed by straight-line programs (under the operational semantics of \cref{tab:apa_rules}) that use statements of type~$i$ at most $j$ times.
We start with a simple lemma that allows us to consider instances of $\APA$ which contain only linearly many statements of type~4.

\smallskip
\begin{restatable}{lemma}{lemsfourbound}\label{lem:sfour_bound}
Wlog, we have $|S_4|\leq  n$.
\end{restatable}
\begin{proof}
Consider any pointer $a\in A$ such that we have statements $\{*a=b_i\}_i$ in $S_4$.
We introduce a new pointer $c$, and 
(i)~we insert a new type-4 statement $*a=c$, and
(ii)~we replace each $*a=b_i$ statement with $c=b_i$.
Performing the above process for each $a$, we create a new instance $(A', S')$ such that for every $a,b\in A$,
we have $b\in \PointsToSet{a}$ in $(A,S)$ iff the same holds in $(A', S')$.
Finally note that $|A'|\leq 2\cdot |A|$ and $|S'|\leq |S|+|A|$, while $|S'_4|\leq |A'|$ as now every pointer appears in the left-hand side of a type~4 statement at most once.

The desired result follows.
\end{proof}

\Paragraph{Algorithm $\BoundedAPAAlgo$.}
We now present our algorithm $\BoundedAPAAlgo$ which solves$(4, j)$ $\APAPA$ for an instance $(A,S)$ and some given $j\geq 0$.
The algorithm performs $j+1$ iterations of $\DReachability{1}$ on graphs $G_{i}=(A, E_i)$, for $i\in[j+1]$, where initially $G_1$ is the Dyck graph in the representation $(G_1, S_4)$ of the $\APA$ instance $(A,S)$.
In iteration $i$, the algorithm solves $\DReachability{1}$ in $G_i$, and then computes all pointers $c$ that flows into some pointer $a$ in $G_i$ for which
there is a statement $*a=b$ in $S_4$.
Then the algorithm resolves the statement by inserting an edge $(b,c,\epsilon)$ in $G_{i+1}$.
See \cref{algo:bounded} for a detailed description.
We conclude this section with the proof of \cref{them:bounded_apa_subcubic}.

\smallskip
\begin{algorithm}
\small
\DontPrintSemicolon
\SetInd{0.3em}{0.3em}
\caption{$\BoundedAPAAlgo$}\label{algo:bounded}
\SetKwProg{Fn}{Function}{:}{}
\KwIn{
An instance $(A,S)$ of $\APA$, a bound $j$ on statements of type~4
}
\KwOut{A set $(a,b)$ of all points-to relationships $b\in \PointsToSet{a}$ witnessed by $(4, j)$-bounded programs.}
\BlankLine
\begin{multicols}{2}
\tcp{Initialization}
\tcp{$|S_4|\leq n$ wlog}
Let $(G_1=(A, E_1), S_4)$ be the Dyck-graph representation of $(A,S)$ \\
Let $V=\{ a_1, a_2\colon a\in A \}$ be a node set\\

\pagebreak
\tcp{Computation}
\ForEach{$i\in [j+1]$}{\label{line:boundeapadalgo_mainloop}
Solve $\DReachability{1}$ in $G_i$ using $\DOneAlgo$\label{line:boundedapaalgo_done}\\
Let $Z_i^{1}=\{ (a_1, b_2)\colon b=\&a \text{ is a statement in }S \}$\\
Let $Z_i^{2}=\{ (a_2, b_2)\colon b \text{ is D-reachable from } a \text{ in } G_i \}$\\
Solve all-pairs reachability in $H_i=(V, Z_i^{1}\cup Z_{i}^{2})$\label{line:boundedapaalgo_tr}\\
Let $E_{i+1}=E_i$\\
\ForEach{statement $*a=b$ in $S_4$}{\label{line:boundedapaalgo_s4}
\ForEach{$c\in A$ with $c_1\Path a_2 $ in $H_i$}{\label{line:boundedapaalgo_path}
Insert $(b,c,\epsilon)$ in $E_{i+1}$\\
}
}
Construct the graph $G_{i+1}=(A, E_{i+1})$
}
\Return{$\{ a,b\colon b_1\Path a_2 \text{ in } H_{j+1} \}$}
\end{multicols}
\end{algorithm}

\begin{proof}[Proof of \cref{them:bounded_apa_subcubic}]
The correctness follows directly from the correctness of $\DOneAlgo$ (\cref{them:d1}).
By induction, at the end of iteration $i$, $\BoundedAPAAlgo$ has solved $(4,i-1)$-bounded $\APAPA$, hence at the end of iteration $j+1$ the algorithm has solved $(4,j)$-bounded $\APAPA$.

We now turn our attention to complexity.
In each iteration of the main loop in \cref{line:boundeapadalgo_mainloop},
we have an invocation to $\DOneAlgo$ in \cref{line:boundedapaalgo_done}, which, by \cref{them:d1}, takes $\Otilde(n^{\omega})$ time.
In addition, the all-pairs reachability in \cref{line:boundedapaalgo_tr}  can be performed using fast BMM~\cite{Munro71} in $O(n^{\omega})$ time.
Finally, by \cref{lem:sfour_bound}, the loop in \cref{line:boundedapaalgo_s4} is executed at most $n$ times, 
while the inner loop in \cref{line:boundedapaalgo_path} is clearly executed at most $n$ times as well. 
Hence, in each iteration, the running time is dominated by the invocation to $\DOneAlgo$, and thus the total time for all iterations is $\Otilde(n^{\omega}\cdot j)$.

The desired result follows.
\end{proof}

\Paragraph{Impact of bounding type-4 statements.}
\cref{them:bounded_apa_subcubic} targets $(4,j)$-bounded $\APAPA$ which limits the number of applications of type-4 statements.
This bounding might miss points-to relationships created by repeatedly nested aliasing.
In practice, the level of indirection is typically small, and thus we expect the above algorithm to be relatively complete.
Given the algorithmic benefits of this approach, an experimental evaluation of its precision is interesting future work.

\section{The Quadratic Hardness of Bounded Andersen's Pointer Analysis}\label{sec:ovhard}

In this section we continue to study bounded $\APA$ and prove \cref{them:ovhardness}, i.e.,
if we restrict our attention to points-to relationships as witnessed by programs of length $\Otilde(1)$,
even the on-demand problem has a quadratic (conditional) lower bound.
Our reduction is from the problem of Orthogonal Vectors~\cite{Williams19}.

\Paragraph{Orthogonal Vectors ($\OV$).}
The input to the problem is two sets $X,Y$, each containing $n'$ vectors in $\{0,1\}^{D}$, for some dimension $D=\omega(\log n')$.
The task is to determine if there exists a pair $(x,y)\in (X\times Y)$ that is orthogonal, i.e., for each $j\in D$, we have $x[j]\cdot y[j] = 0$.
The respective hypothesis states that the problem cannot be solved in time $O(n'^{2-\epsilon})$, for any fixed $\epsilon>0$.

\begin{figure}
\newcommand{\xdisposition}{3}
\newcommand{\ydisposition}{-3.5}
\newcommand{\xtstep}{0.9}
\newcommand{\ytstep}{0.5}
\newcommand{\xstep}{1.2}
\newcommand{\ystep}{-1.5}
\newcommand{\gatex}{0.8}
\newcommand{\gatey}{0.4}

\centering
\begin{tikzpicture}[thick,
pre/.style={<-,shorten >= 2pt, shorten <=2pt,  thick},
post/.style={->, thick},
seqtrace/.style={->, line width=2},
node/.style={circle},
postpath/.style={->, thick},
und/.style={very thick, draw=gray},
event/.style={rectangle, minimum height=0.8mm, minimum width=3mm, fill=black!100,  line width=1pt, inner sep=0},
virt/.style={circle,draw=black!50,fill=black!20, opacity=0}
]

\begin{scope}[shift={(4.75*\xstep, 1*\ydisposition + 1*\ystep)}]
\node[] at (2*\xstep, 0*\ystep){
$y_{1}=
\begin{bmatrix}
1 \\
0
\end{bmatrix}
$
};

\node[] at(-1.5*\xstep, 0*\ystep)	{
$\begin{aligned}
*u_1^{1}=& v_1^{1}
\end{aligned}$
};

\end{scope}

\begin{scope}[shift={(4.75*\xstep, 2*\ydisposition + 1*\ystep)}]
\node[] at (2*\xstep, 0*\ystep){
$y_{2}=
\begin{bmatrix}
0\\
1
\end{bmatrix}
$
};

\node[] at(-1.5*\xstep, 0*\ystep)	{
$\begin{aligned}
*u_1^{2}=& v_1^{2}
\end{aligned}$
};

\end{scope}

\begin{scope}[shift={(-0.75*\xstep, 1*\ydisposition + 1*\ystep)}]
\node[] at (-2*\xstep, 0*\ystep){
$x_{1}=
\begin{bmatrix}
1 \\
1
\end{bmatrix}
$
};

\node[] at(2*\xstep, 0*\ystep)	{
$\begin{aligned}
*a_2^{1}=& b_2^{1}
\end{aligned}$
};

\end{scope}

\begin{scope}[shift={(-0.75*\xstep, 2*\ydisposition + 1*\ystep)}]
\node[] at (-2*\xstep, 0*\ystep){
$x_{2}=
\begin{bmatrix}
0\\
1
\end{bmatrix}
$
};

\node[] at(2*\xstep, 0*\ystep)	{
$\begin{aligned}
*a_2^{2}=& b_2^{2}
\end{aligned}$
};

\end{scope}

\node[] (z) at (2*\xstep, 2*\ydisposition - 0.8*\ystep) {$z$};
\node[] (s) at (-1.5*\xstep, 2*\ydisposition - 0.8*\ystep) {$s$};
\node[] (t) at (5.5*\xstep, 2*\ydisposition - 0.8*\ystep) {$t$};
\foreach \v in {1,2}{
\begin{scope}[shift={(0*\xdisposition, \v*\ydisposition)}]
\node[] (x\v1) at (0*\xstep, 0*\ystep) {$a^{\v}_{1}$};
\node[] (x\v2) at (0*\xstep, 1*\ystep) {$a^{\v}_{2}$};
\node[] (y\v2) at (0*\xstep, 1.75*\ystep) {$b^{\v}_{2}$};

\node[] (u\v1) at (4*\xstep, 0*\ystep) {$u^{\v}_{1}$};
\node[] (u\v2) at (4*\xstep, 1.75*\ystep) {$u^{\v}_{2}$};
\node[] (v\v1) at (4*\xstep, 1*\ystep) {$v^{\v}_{1}$};

\draw[postpath, bend left=0] (x\v1) to node[left]{$*$} (x\v2);
\end{scope}
}

\draw[postpath, out=0, in=90, looseness=1.5] (x11) to node[above, pos=0.2, sloped]{$\&$} (z);
\draw[postpath, bend left=10] (x21) to node[above, midway, sloped]{$2\&$} (z);
\draw[postpath, bend right=10] (x21) to node[below, midway, sloped]{$\&$} (z);

\draw[postpath, out=90, in=180, looseness=1.5] (z) to node[above, pos=0.8, sloped]{$*$} (u11);
\draw[postpath, bend right=10] (z) to node[below, midway, sloped]{$*$} (u21);
\draw[postpath, bend left=10] (z) to node[below, midway, sloped]{} (u21);

\draw[postpath, bend left=0] (u22) to node[right]{$\&$} (v21);

\draw[postpath, bend left=10] (u12) to node[left]{$\&$} (v11);
\draw[postpath, bend right=10] (u12) to node[right]{$2\&$} (v11);

\draw[postpath, bend left=10] (s) to node[above, midway, sloped]{$\&$} (y12);
\draw[postpath, bend right=30] (s) to node[below, midway, sloped]{$\&$} (y22);
\draw[postpath, bend left=10] (u12) to  (t);
\draw[postpath, bend right=30] (u22) to  (t);

\end{tikzpicture}
\caption{
Reduction from $\OV$ with vector sets $A=\{x_1, x_2\}$ and $B=\{y_1, y_2\}$ to $(\All, O(\log n))$-bounded $\APA$
on the pair $s\in\PointsToSet{t}$.
}
\label{fig:ov_hard}
\end{figure}
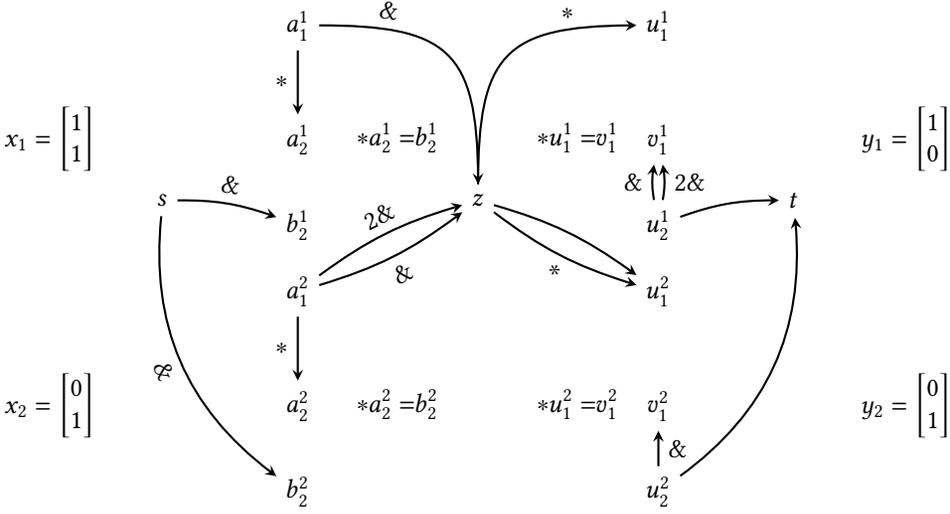
\Paragraph{Reduction from $\OV$.}
Consider an instance $X,Y$ of $\OV$, and assume wlog that $D$ is even.
Here we show how to construct a $\SPAPA$ instance with two distinguished pointers $s$ and $t$ such that $s\in\PointsToSet{t}$ iff there exists an orthogonal pair of vectors in $X\times Y$.

\SubParagraph{Intuition.}
We start with a high-level intuition, while \cref{fig:ov_hard} provides an illustration.
Consider the first two vectors $x_1\in Y$ and $y_1\in Y$, and focus on the first coordinate.
Recall the representation of $\APA$ as a Dyck graph $G$ and a set of type-4 statements $S_4$, and let $\ov{G}$ be the resolved Dyck graph.
We introduce two pointers $a_1^1$ and $u_1^1$, with the goal that $a_1^1$ flows into $u_1^1$ in the resolved Dyck graph $\ov{G}$ iff $x_1[1]\cdot y_1[1]=0$, i.e., $x_1$ and $y_1$ are orthogonal as far as the first coordinate is concerned.
We achieve this by introducing a distinguished node $z$, and having two edges $a_1^1\DEdge{\&}z$ and $z\DEdge{*}u_1^1$. 
Moreover, if $x_1[1]=0$, we also create a path $P_1^x\colon a_1^1\DPath{\&\&} z$.
Similarly, if $y_1[1]=0$, we also create a path $P_1^y\colon z\DPath{\epsilon}u_1^1$ (i.e., $P_1^y$ is simply an $\epsilon$-labeled edge).
Observe that $a_1^1$ flows into $u_1^1$ iff $x_1[1]\cdot y_1[1]=0$.
If $x_1[1]\cdot y_1[1]=1$ then nothing happens and the process stops here.
Otherwise, $x_1$ and $y_1$ are potentially orthogonal, so we proceed with the second coordinate.
We create a node $v_1^1$ and a type-4 statement $*u_1^1=v_1^1$; since $a_1^1$ flows into $u_1^1$, the resolved graph $\ov{G}$ also has an edge $v_1^1\DEdge{\epsilon} u_1^1$.
The contents of $x_1$ and $y_1$ on the second coordinate are encoded via paths to new pointers $a_2^1$ and $u_2^1$, respectively.
In particular, we have two edges $a_1^1\DEdge{*}a_2^1$ and $u_2^1\DEdge{\&} v_1^1$.
Moreover, if $x_1[2]=0$, we also create a path $P_2^x\colon a_1^1\DPath{\epsilon} a_2^1$ (i.e., $P_2^x$ is simply an $\epsilon$-labeled edge).
Similarly, if $y_1[2]=0$, we also create a path $P_2^y\colon u_2^1\DPath{\&
\&} v_1^1$.
Observe that $u_1^2$ flows into $a_2^1$ iff $v_1^1\DEdge{\epsilon} a_1^1$ (which is established by the previous step, as $x_1[1]\cdot y_1[1]=0$) and also $x_1[2]\cdot y_1[2]=0$, thereby establishing that $x_1$ and $y_1$ appear orthogonal on the first two coordinates.
In that case, we have another pointer $b_2^1$ and type-4 statement $*a_2^1=b_2^1$, which inserts an edge $b_2^1\DEdge{\epsilon}u_2^1$ in $\ov{G}$.
From here on, the process repeats as before, with $b_2^1$ playing the role of pointer $z$ initially, and the contents of $x_1[3]$ captured in paths from $a_3^1\DPath{} b_2^1$, and the contents of $y_1[3]$ captured in paths $u_2^1\DPath{}u_3^1$.
In the end, we have that $u_D^1$ is D-reachable from $b_D^1$ iff $x_1[j]\cdot y_1[j]=0$ for all $j\in [D]$.

Finally, to capture all potential pairs of vectors $x_{i_1}, y_{i_2}\in X\times Y$, 
we connect the pointer $z$ to all pointers $a_1^{i_1}$ and $u_1^{i_2}$, as above.
To complete the reduction, we make $s\DEdge{\&} b_D^{i_1}$ and $u_D^{i_2}\DEdge{\epsilon} t$, and thus $s$ flows into $t$ via the D-reachable path $u_D^{i_1}\DPath{} b_D^{i_2}$ iff $x_{i_1}$ and $y_{i_2}$ are orthogonal.

\SubParagraph{Formal construction}
We now proceed with the formal construction, as follows.
First, we introduce a pointer $z$.

For every vector $x^{i}\in X$, we introduce pointers $a_1^{i},\dots, a_{D}^{i}$ and $b_2^{i},b_4^{i},\dots, b_{D}^{i}$.
\begin{compactenum}
\item We have $z=\&a_1^i$. If $x^i[1]=0$, we also introduce a new pointer $\hat{a}_1^{i}$ and two assignments
$z=\& \hat{a}_1^{i}$ and $\hat{a}_1^1=\&a_1^{i}$.
\item For every even $j\in [D]$, we have $*a_j^{i}=b_j^{i}$ and $a_j^{i}=*a_{j-1}^{i}$.
If $x^i[j]=0$, we also have $a_{j}^{i}=a_{j-1}^{i}$.
\item For every odd $j\in [D]$ with $j>1$, we have $b_{j-1}^{i}=\& a_j^{i}$.
If $x^i[j]=0$, we also introduce a new pointer $\hat{a}_j^{i}$ and two assignments
$b_{j-1}^{i}=\& \hat{a}_j^{i}$ and $\hat{a}_j^i=\&a_j^{i}$.
\end{compactenum}
For every vector $y^{i}\in Y$, we introduce pointers $u_1^i,\dots, u_{D}^i$ and $v_1^{i}, v_3^{i},\dots, v_{D-1}^{i}$.
\begin{compactenum}
\item We have $u_1^{i}=*z$. If $y^i[1]=0$, we also have $u_1^i=z$.
\item For every odd $j\in [D]$, we have $*u_{j}^{i}=v_j^i$ and $u_j^i=*u_{j-1}^i$.
If $y^i[j]=0$, we also have $u_j^i=u_{j-1}^i$.
\item For every even $j\in[D]$, we have $v_{j-1}^i=\& u_j^i$.
If $y^i[j]=0$, we also introduce a new pointer $\hat{u}_j^i$ and two assignments
$v_{j-1}^i=\& \hat{u}_j^i$ and $\hat{u}_j^i=\& u_j^i$.
\end{compactenum}
Finally, we introduce two pointers $s$ and $t$.
For every $i\in [n']$, we have $b_{D}^i=\& s$ and $t=u_{D}^i$.
The on-demand question is whether $s\in \PointsToSet{t}$.
Observe that we have used $n=O(n'\cdot D)$ pointers, and the above construction can be easily carried out in $O(n)$ time.

\Paragraph{Correctness.}
We now establish the correctness of the above construction.
The key idea is as follows.
Recall our definition of the Dyck-graph representation $(G=(A,E),S_4)$ of the $\APA$ instance, and the resolved Dyck graph $\ov{G}$ (see \cref{subsec:dyck_reachability}).
The resolved graph $\ov{G}$ is constructed from $G$ by iteratively
(i)~finding three pointers $a,b,c$ such that $a$ flows into $b$ and we have a type 4 statement $*b=c$, and
(ii)~inserting an edge $c\DEdge{\epsilon} a$ in $G$.
The above construction guarantees that, for two integers $i_1, i_2\in [n']$, the following hold by induction on $j\in [D]$.
\begin{compactenum}
\item If $j$ is odd,
we have $v^{i_2}_{j}\DEdge{\epsilon} a^{i_1}_{j}$ iff $\sum_{j'\leq j}x^{i_1}[j']\cdot y^{i_2}[j']=0$.
\item If $j$ is even,
we have $b^{i_1}_{j}\DEdge{\epsilon} u^{i_2}_{j}$ iff $\sum_{j'\leq j}x^{i_1}[j']\cdot y^{i_2}[j']=0$.
\end{compactenum}
Once such an $\epsilon$-labeled edge is inserted for some $j$, it creates a path that leads to a flows-into relationship that leads to inserting the next $\epsilon$-labeled edge for $j+1$ iff $x^{i_1}[j+1]\cdot y^{i_2}[j+1]=0$.
Note that $s$ flows into $t$ iff there exist $i_1, i_2\in [n']$ such that $b_{i_1}^{D}\DEdge{\epsilon} u_{i_2}^{D}$, which, by the above, holds iff $x^{i_1}$ and $y^{i_2}$ are orthogonal.
Finally, since $D=\Theta(\log n)$, the witness program for $s\in \PointsToSet{t}$ has length $\Otilde(1)$.

The above idea is formally captured in the following two lemmas

\smallskip
\begin{restatable}{lemma}{lemovcompleteness}\label{lem:ov_completeness}
If there exist $i_1, i_2\in [n']$ such that $x^{i_1}$ and $y^{i_2}$ are orthogonal, then $s$ flows into $t$ in $\ov{G}$.
Moreover, there exists a witness program $\Program$ of length $O(\log n)$ that results in $s\in \Store{t}$.
\end{restatable}

\smallskip
\begin{restatable}{lemma}{lemovsoundness}\label{lem:ov_soundness}
If $s$ flows into $t$ in $\ov{G}$, there exist $i_1, i_2\in [n']$ such that $x^{i_1}$ and $y^{i_2}$ are orthogonal.
\end{restatable}

We conclude this section with the proof of \cref{them:ovhardness}.

\begin{proof}[Proof of \cref{them:ovhardness}.]
\cref{lem:ov_completeness} and \cref{lem:ov_soundness}, together with \cref{lem:resolved_graph}, state the correctness of the reduction.
Note that the $\APA$ instance we constructed has $n=O(n'\cdot D)$ pointers and $m=O(n'\cdot D)$ assignments,
hence it is a sparse instance.
Moreover the time for the construction is $O(n'\cdot D)$.
Assume that there exists some fixed $\epsilon >0$ such that $\SPAPA$ can be solved in $O(n^{2-\epsilon})$ time.
Then we have a solution for the $\OV$ instance in time $O((n'\cdot D)^{2-\epsilon })$ time, which violates the Orthogonal-Vectors hypothesis.

The desired result follows.
\end{proof}

\section{The Parallelizability  of Andersen's Pointer Analysis}\label{sec:pcomplete}

In this section we address the parallelizability of $\APA$, and show the following results.
In \cref{subsec:pcomplete} we prove \cref{them:pcomplete}, which shows that $\APA$ is not parallelizable under standard hypotheses in complexity theory.
In \cref{subsec:parallelizable} we prove \cref{them:bounded_parallelizable}, which shows that bounded $\APA$ is efficiently parallelizable as long as we focus on poly-logarithmically (i.e., $O(\log^c n)$, for some constant $c$) many applications of type~4 statements.

\subsection{Andersen's Pointer Analysis is not Parallelizable}\label{subsec:pcomplete}
In this section we prove \cref{them:pcomplete}, i.e., that $\SPAPA$ is P-complete.
This implies that the problem is unlikely to be parallelizable.
Our reduction is from $\MonotoneCVP$.

\Paragraph{The problem $\MonotoneCVP$.}
The input to the problem is a circuit represented as a sequence of assignments $(A_1, \dots, A_{n'})$, such that for all $i\in [n']$, $A_i$ has one of the following types, where $j<k<i$.
\begin{align*}
A_i=0 \qquad\qquad A_i=1 \qquad\qquad A_i=A_j \land A_k \qquad\qquad A_i=A_j \lor A_k
\end{align*}
The condition $j<k<i$ ensures acyclicity.
The assignments $A_i=0$ and $A_i=1$ are the \emph{inputs}, whereas all other assignments are the \emph{gates}.
The task is to compute whether $A_{n'}$ evaluates to $1$ under the standard Boolean algebra interpretation of the operators $\land$ and $\lor$.
It is known that $\MonotoneCVP$ is P-complete even when every assignment has fan-out 2 (except $A_{n'}$ and inputs $A_i$)~\cite{Greenlaw95}.

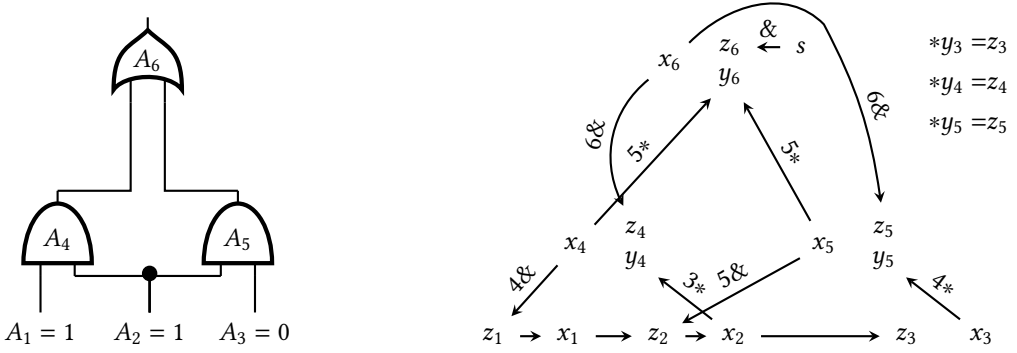
\begin{figure}
\newcommand{\xdisposition}{0}
\newcommand{\ydisposition}{0}
\newcommand{\xtstep}{0.9}
\newcommand{\ytstep}{0.6}
\newcommand{\xstep}{1.1}
\newcommand{\ystep}{1.2}
\newcommand{\gatex}{0.8}
\newcommand{\gatey}{0.4}

\newcommand{\drawGate}[3]{
\node[node] (T#1) at (#2 + \gatex, #3 + \gatey/2)  {$z_{#1}$};
\node[node] (#1bar) at (#2 + \gatex, #3 - \gatey/2)  {$y_{#1}$};
\node[node] (#1) at (#2,#3)  {$x_{#1}$};
}
\newcommand{\drawInput}[3]{
\node[node] (T#1) at (#2, #3 )  {$z_{#1}$};
\node[node] (#1) at (#2+ 1.25*\gatex,#3)  {$x_{#1}$};
}
\centering
\begin{tikzpicture}[thick,
pre/.style={<-,shorten >= 2pt, shorten <=2pt,  thick},
post/.style={->, thick},
node/.style={circle},
seqtrace/.style={->, line width=2},
postpath/.style={->, thick},
und/.style={very thick, draw=gray},
event/.style={rectangle, minimum height=0.8mm, minimum width=3mm, fill=black!100,  line width=1pt, inner sep=0},
virt/.style={circle,draw=black!50,fill=black!20, opacity=0}
]

\drawInput{1}{0*\xstep}{0*\ystep}
\drawInput{2}{2*\xstep}{0*\ystep}
\drawInput{3}{5*\xstep}{0*\ystep}
\drawGate{4}{1*\xstep}{1*\ystep}
\drawGate{5}{4*\xstep}{1*\ystep}
\drawGate{6}{2.5*\xstep-\gatex/2}{3*\ystep}
\node[node] (s) at (3*\xstep+ \gatex, 3*\ystep+\gatey/2) {$s$};
\draw[post] (s) tonode[midway, above, sloped]{$\&$} (T6);
\draw[postpath] (4) to node[midway, above, sloped]{$4\&$} (T1);
\draw[postpath] (2) to node[midway, above, sloped]{$3*$} (4bar);
\draw[post] (1) to (T2);
\draw[postpath] (5) to node[midway, above, sloped]{$5\&$} (T2);
\draw[postpath] (3) to node[midway, above, sloped]{$4*$} (5bar);
\draw[post] (2) to  (T3);
\draw[postpath, bend right=40] (6) to node[midway, above, sloped]{$6\&$} (T4);
\draw[postpath, smooth, looseness=1] (6) to ($ (s) + (+0.3,0.3) $)  to[bend left=10]  node[midway, above, sloped]{$6\&$} (T5);
\draw[postpath] (4) to node[midway, above, sloped]{$5*$} (6bar);
\draw[postpath] (5) to node[midway, above, sloped]{$5*$} (6bar);

\draw[post, ] (T1) to (1);
\draw[post, ] (T2) to (2);

\node[] at(5.75*\xstep, 2.75*\ystep)	{
$\begin{aligned}
*y_3=& z_3\\
*y_4=& z_4\\
*y_5=& z_5
\end{aligned}$
};

\begin{scope}[shift={(-6*\xstep,0)}]

\renewcommand{\xstep}{0.8}
\node[and port, draw, thick,  rotate=90, label={[label distance=-8mm]$A_4$}, logic gate inputs=nn] (A4) at (1*\xstep, 1.5*\ystep)  {};
\node[and port, draw, thick,  rotate=90, label={[label distance=-8mm]$A_5$}, logic gate inputs=nn] (A5) at (4*\xstep, 1.5*\ystep)  {};
\node[or port, draw, thick,  rotate=90, label={[label distance=-7.25mm]$A_6$}, logic gate inputs=nn] (A6) at (2.5*\xstep, 3.4*\ystep)  {};

\node[below= \ystep/2.5 of A4.in 1](A1)  {$A_1=1$};
\node[below right= \ystep/2.5 and \xstep/2 of A4.in 2](A2)  {$A_2=1$};
\node[below = \ystep/2.5 of A5.in 2](A3)  {$A_3=0$};

\draw[thick] (A1) |- (A4.in 1);
\draw[thick] (A2) |- (A4.in 2);
\draw[thick] (A2) |- (A5.in 1);
\draw[thick] (A3) |- (A5.in 2);

\draw[thick, -*] (A2) to ($ (A2) + (0,\ystep/1.35) $) ;

\draw[thick] (A4.out) -| (A6.in 1);
\draw[thick] (A5.out) -| (A6.in 2);

\end{scope}

\end{tikzpicture}
\caption{
A $\MonotoneCVP$ instance (left) and the corresponding $\SPAPA$ instance for $s\in\PointsToSet{x_6}?$ (right).
Squiggly arrows represent paths of unique nodes, where the path has the corresponding path label.
}
\label{fig:p_complete}
\end{figure}

\Paragraph{Reduction from $\MonotoneCVP$.}
Consider an instance $(A_1, \dots, A_{n'})$ of $\MonotoneCVP$, and we construct an instance $(A,S)$ of $\SPAPA$.
An illustration is given in \cref{fig:p_complete}.
\begin{compactenum}
\item For every input $A_i$, we introduce two pointers $x_i, z_i\in A$.
We have a statement $x_i=z_i$ in $S$ iff $A_i=1$.
\item For every gate $A_i=A_j\land A_k$, with $j<k<i$, we introduce
(i)~three pointers $x_i, y_i, z_i$,
(ii)~$i$ pointers $x_i^1,\dots, x_i^i$, and
(iii)~$i-1$ pointers $y_i^1,\dots, y_i^{i-1}$.
We have the following statements in $S$.
\begin{align*}
&x_i^1=\&x_i   \qquad x_i^2= \&x_i^1  \qquad \cdots \qquad x_i^i=\&x_i^{i-1} \qquad z_j=\&x_i^{i} \qquad\qquad \text{and}\\
&y_i^{i-1}=*y_k \qquad  y_i^{i-2}= *y_i^{i-1} \qquad \cdots \qquad    y_i^{1}= *y_i^{2} \qquad  y_i=*y_i^{1}
\end{align*}
For every gate $A_i=A_j\lor A_k$, with $j<k<i$, we introduce
(i)~three pointers $x_i, y_i, z_i$,
(ii)~$i$ pointers $x_i^1,\dots, x_i^i$ and $i$ pointers $\hat{x}_i^1,\dots, \hat{x}_i^i$, and
(iii)~$i-1$ pointers $y_i^1,\dots, y_i^{i-1}$ and $i-1$ pointers $\hat{y}_i^1,\dots, \hat{y}_i^{i-1}$.
We have the following statements in $S$.
\begin{align*}
&x_i^1=\&x_i   \qquad x_i^2= \&x_i^1  \qquad \cdots \qquad x_i^i=\&x_i^{i-1} \qquad z_j=\&x_i^{i} \qquad\qquad \text{and}\\
&\hat{x}_i^1=\&x_i   \qquad \hat{x}_i^2= \&\hat{x}_i^1  \qquad \cdots \qquad \hat{x}_i^i=\&\hat{x}_i^{i-1} \qquad z_k=\&\hat{x}_i^{i} \qquad\qquad \text{and}\\
&y_i^{i-1}=*y_j \qquad  y_i^{i-2}= *y_i^{i-1} \qquad \cdots \qquad    y_i^{1}= *y_i^{2} \qquad  y_i=*y_i^{1} \qquad\qquad \text{and}\\
&\hat{y}_i^{i-1}=*\hat{y}_k \qquad  \hat{y}_i^{i-2}= *\hat{y}_i^{i-1} \qquad \cdots \qquad    \hat{y}_i^{1}= *\hat{y}_i^{2} \qquad  y_i=*\hat{y}_i^{1}
\end{align*}
\item For every gate $A_i$, we add a statement $*y_i=z_i$.
\item Finally, we introduce a special pointer $s$, and a statement $z_{n'}=\&s$. The on-demand question is whether $s\in\PointsToSet{x_{n'}}$.
\end{compactenum}
Note that the size of our $\APA$ instance is $\Theta(n'^2)$ i.e., the construction leads to a quadratic blow-up.
Nevertheless, it is easy to carry out the above construction in logarithmic space, as required for a log-space reduction.

\Paragraph{Correctness.}
We now establish the correctness of the reduction.
We first present an intuitive insight and then the formal steps of the proof.

\SubParagraph{Intuitive argument of correctness.}
Let $(G=(A, E), S_4)$ be the Dyck-graph representation of the constructed $\APA$ instance $(A, S)$, and let $\ov{G}$ be the resolved Dyck graph (see \cref{subsec:dyck_reachability}).
Recall that the resolved graph $\ov{G}$ is constructed from $G$ by iteratively
(i)~finding three pointers $a,b,c$ such that $a$ flows into $b$ and we have a type 4 statement $*b=c$, and
(ii)~inserting an edge $c\DEdge{\epsilon} a$ in $G$.

The correctness of the construction relies on the following invariant.
For every $i$ such that $A_i$ is a gate with inputs $A_j$, $A_k$, we have that $x_i$ is D-reachable from $z_i$ iff
(i)~$z_j\DEdge{\epsilon} x_j$ and $z_k\DEdge{\epsilon}x_k$ (if $A_i$ is an AND gate), or
(ii)~$z_j\DEdge{\epsilon} x_j$ or $z_k\DEdge{\epsilon}x_k$ (if $A_i$ is an OR gate).
Observe that in $\ov{G}$ we have potentially other nodes $u$ that are D-reachable from $z_i$.
For example, in \cref{fig:p_complete}, we have $z_4\DEdge{\epsilon} x_5^1$ in $\ov{G}$,
where $x_5^1$ is the second node in the path $x_5\DPath{5\&}z_2$ (not explicitly shown in the figure).
This occurs because we have a statement $*y_4=z_4$, and $x_5^1$ flows into $y_4$.
The key insight is then that the invariant holds despite such ``unwanted'' edges.
Given the invariant, the correctness proof follows by an induction on the depth of the circuit.

\SubParagraph{Formal Correctness.}
We now make the above insights formal.
We start with a technical lemma, which shows that if $x_i$ is D-reachable from $z_i$ in $\ov{G}$, then, in fact, $z_i\to x_i$,
as $x_i$ flows into $y_i$ and we have a statement $*y_i=z_i$.

\smallskip
\begin{restatable}{lemma}{lempcompletelemma}\label{lem:p_complete_lemma}
For all $i\in [n']$, $x_i$ is D-reachable from $z_i$ in $\ov{G}$ iff $x_i$ flows into $y_i$ in $\ov{G}$.
\end{restatable}

The following lemma establishes our main invariant, which relates the encoding of the output of a gate in $\ov{G}$ to the encoding of its inputs.
\smallskip
\begin{restatable}{lemma}{lempcompleteconstruction}\label{lem:pcomplete_construction}
Consider a gate $A_i$. We have that $x_i$ is D-reachable from $z_i$ in $\ov{G}$ iff
\begin{compactenum}
\item $A_i$ is an AND gate $A_i=A_j\land A_k$, and $x_j$ is D-reachable from $z_j$ and $x_k$ is D-reachable from $z_k$, or
\item $A_i$ is an OR gate $A_i=A_j\lor A_k$, and $x_j$ is D-reachable from $z_j$ or $x_k$ is D-reachable from $z_k$.
\end{compactenum}
\end{restatable}

The following lemma establishes the correctness of the construction.
Its proof follows by an induction on the depth of the circuit, and using \cref{lem:pcomplete_construction} on the respective gate.

\smallskip
\begin{restatable}{lemma}{lempcompletecorrectness}\label{lem:pcomplete_correctness}
We have that $x_{n'}$ is D-reachable from $z_{n'}$ iff $A_{n'}$ evaluates to $1$.
\end{restatable}

We conclude this section with the proof of \cref{them:pcomplete}.

\begin{proof}[Proof of \cref{them:pcomplete}]
Membership in P is known (e.g., \cref{them:cubic_upper_bound}), so we need to argue that the problem is P-hard.
Observe that $s$ flows into $x_{n'}$ iff $x_{n'}$ is D-reachable from $z_{n'}$.
By \cref{lem:pcomplete_correctness}, we have that $s$ flows into $x_{n'}$ iff $A_{n'}$ evaluates to $1$.
By \cref{lem:resolved_graph}, we have that $s\in \PointsToSet{x_{n'}}$ iff $A_{n'}$ evaluates to $1$.
The desired result follows.
\end{proof}

\subsection{Bounded Andersen's Pointer Analysis is Parallelizable}\label{subsec:parallelizable}

Finally, in this section we develop an algorithm for solving bounded $\APAPA$, and thus prove \cref{them:bounded_parallelizable}.

\Paragraph{Parallel bounded $\APA$.}
The algorithm is a parallelization of $\BoundedAPAAlgo$ (\cref{algo:bounded}) for sequential bounded $\APA$.
The parallel algorithm performs the iterations of the main loop of \cref{line:boundeapadalgo_mainloop} sequentially,
while the body of the loop is run in parallel.
We outline the steps of the parallelization.
\begin{compactenum}
\item In \cref{line:boundedapaalgo_done}, $\BoundedAPAAlgo$ invokes the routine $\DOneAlgo$ for computing $\DReachability{1}$ (\cref{algo:done}).
Recall, by \cref{lem:d1_distance}, that the distance between two reachable nodes is $O(n^2)$, and hence the same bound holds for the maximum stack height of the shortest path that witnesses reachability between such nodes.
It follows that $\DReachability{1}$ can be reduced to standard graph reachability on a graph with $O(n^3)$ nodes of the form $(u,i)$, where $u$ is a node of the Dyck graph and $i\in [O(n^2)]$ encodes the stack height. 
Finally, graph reachability is solved in parallel~\cite{Papadimitriou93}.
\item In \cref{line:boundedapaalgo_tr}, $\BoundedAPAAlgo$ constructs another graph $H$ and solves all-pairs reachability on $H$.
Now, we perform the construction of $H$ in parallel, while all-pairs reachability in $H$ is also computed in parallel, as in the previous item.
\item Finally, we execute the two nested loops in \cref{line:boundedapaalgo_s4} and \cref{line:boundedapaalgo_path} in parallel, by using one processor per triplet of pointers $(a,b, c)$.
\end{compactenum}

The correctness of the parallelization follows directly from the correctness of the sequential version (\cref{them:bounded_apa_subcubic}).
We conclude with the complexity analysis, which establishes \cref{them:bounded_parallelizable}.

\begin{proof}[Proof of \cref{them:bounded_parallelizable}]
In each iteration, the parallel running time is the time required to compute the transitive closure of a graph of $O(n^3)$ nodes.
Since graph reachability is in $\NC^2$~\cite{Papadimitriou93}, each iteration takes $O(\log^2 n)$ time.
Hence, executing the main loop sequentially for $j$ iterations  yields $O(j\cdot \log^2 n)$ time.
In particular, for $j=\log^{i} n$ iterations, the algorithm solves $O(4,\log^i n)$-bounded $\APAPA$ in $O(\log^{i+2} n)$ parallel time,
hence the problem is in NC$^{i+2}$.

The desired result follows.
\end{proof}

\section{Conclusion}\label{sec:conclusion}
Andersen's Pointer Analysis is a standard approach to static, flow-insensitive pointer analysis.
Despite its long history and practical importance, the complexity of the analysis had remained illusive.
In this work, we have drawn a rich fine-grained and parallel complexity landscape based on various aspects of the problem.
We have shown that even deciding whether a single pointer may point to a specific heap location is unlikely to have sub-cubic complexity,
and additionally, the problem is not parallelizable.
These results strongly characterize the hardness of the problem. 
On the positive side, we have presented a bounded version of the problem that becomes solvable in nearly matrix-multiplication time,
and have established a conditional quadratic lower bound for the bounded version of the problem.

Our positive results build some stable ground for further practical improvements for Andersen's pointer analysis.
We expect that our solution to the bounded version of the problem, which essentially reduces to a small number of standard transitive closure operations,
can provide the basis for faster practical approaches:~graph transitive closure solvers have been heavily optimized over the years in both software and hardware, while the full parallelizability of the bounded problem opens itself up to more efficient multithreaded implementations.
As our focus in this work has been on characterizing the tractability landscape of the problem, we have left the practical realizations of our results for interesting future work.

\begin{acks}
We wish to thank Phillip G. Bradford for bringing to our attention his result on $\DReachability{1}$ and for helpful comments in comparing his algorithm to ours, as well as anonymous reviewers for their constructive feedback in an earlier version of this manuscript.
\end{acks}


\bibliography{bibliography}

\clearpage
\appendix
\section{Proofs}\label{appsec:proofs}

\subsection{Proofs of \cref{sec:prelim}}\label{appsubsec:proof_prelim}

\smallskip
\lemapatodyck*
\begin{proof}
We prove each direction separately.

\noindent{$(\Rightarrow)$.}
Assume that $b\in \PointsToSet{a}$, and we argue that there exists a pointer $c$ such that
(i)~$b\DEdge{\&}c$ and
(ii)~$a$ is D-reachable from $c$,
and hence $b$ flows into $a$.
We employ the operational semantics of $\APA$.
Let $\Program$ be a minimal program that results in $b\in \Store{a}$.
The proof is by induction on the length of $\Program$.
For the base case, we have $|\Program|=1$. Then $\Program$ consists of a single statement $a=\&b$, hence the lemma holds for $a=c$.
Now let $|\Program|=\ell+1$, and by the induction hypothesis the statement holds for all points-to relationships witnessed by programs $\Program'$ with length $|\Program'|=\ell$.
We distinguish the last statement $s$ of $\Program$.
Note that $s$ is either a type~1 or a type~3 statement.
\begin{compactenum}
\item $s$ is of the form $a=d$. By the induction hypothesis on $\Program'$, we have a node $c'$ such that 
(i)~$b\DEdge{\&}c'$ and
(ii)~$d$ is D-reachable from $c'$.
But then $a$ is also D-reachable from $c$, and hence the lemma holds for $c=c'$.
\item $s$ is of the form $a=*d$. Then there exists a pointer $e$ such that $\Program'$ witnesses $e\in \Store{d}$ and $b\in \Store{e}$.
By the induction hypothesis on $e\in \Store{d}$, there exists a node $c_1$ such that
(i)~$e\DEdge{\&} c_1$ and
(ii)~$d$ is D-reachable from $c_1$.
Note that this implies that $a$ is D-reachable from $e$.
Similarly, by the induction hypothesis on $b\in \Store{e}$, there exists a node $c_2$ such that
(i)~$b\DEdge{\&} c_2$ and
(ii)~$e$ is D-reachable from $c_2$.
It follows that $a$ is D-reachable from $c_2$, hence the lemma holds for $c=c_2$.
\end{compactenum}

\noindent{$(\Leftarrow)$.}
Assume that $b$ flows into $a$, hence, there exists a pointer $c$ such that 
(i)~$b\DEdge{\&}c$, and
(ii)~$a$ is D-reachable from $c$,
and we argue that $b\in \PointsToSet{a}$.
The proof is by induction on the label $\Label(P)$ of the path $P\colon c\Path a$ that witnesses the Dyck-reachability.
Note that, by construction, we have a statement $c=\&b$, and the statement holds if $P$ has no edges.
\begin{compactenum}
\item $\Label(P)=\epsilon$.
Then we have a statement $a=c$, which, together with $c=\&b$, implies that $c\in \PointsToSet{a}$.
\item $\Label(P)=\&  \StartNonTerminal *$.
Then there exist  intermediate nodes $d_1$, $d_2$,  such that $P$ can be decomposed as 
$c\DEdge{\&} d_1\DPath{\StartNonTerminal}d_2  \DEdge{*} a$.
Let $P'$ be the intermediate $d_1\DPath{\StartNonTerminal}d_2$ path,
and by the induction hypothesis, we have $c\in \PointsToSet{d_2}$.
By construction, we have a statement $a=*d_2$, and thus $b\in \PointsToSet{b}$.
\item $\Label(P)=\StartNonTerminal \StartNonTerminal$.
Then there exists an intermediate node $d$ such that $P$ can be decomposed as $P\colon P_1\circ P_2$,
where $P_1\colon c\DPath{\StartNonTerminal} d$ and $P_2\colon d\DPath{\StartNonTerminal} a$.
By the induction hypothesis on $P_1$, we have $b\in \PointsToSet{d}$.
Then, by the induction hypothesis on $P_2$, we have $b\in \PointsToSet{a}$.
\end{compactenum}

The desired result follows.
\end{proof}

\smallskip
\lemresolvedgraph*
\begin{proof}
Let $G_1,\dots, G_{\ell}$ be a sequence of Dyck graphs where $G_1=G$ and $(G=(V,E),S_4)$ is the Dyck-graph representation of $(A,S)$, and $G_{i+1}$ is constructed from $G_i$ by 
\begin{compactenum}
\item identifying all nodes $c$ that flow into some node $a$ for which there is a statement $*a=b$, and
\item inserting an edge $(b,c,\epsilon)$ in $G_{i+1}$.
\end{compactenum}
Clearly this sequence is finite, and $G_{\ell}=\ov{G}$.
It is straightforward to establish by induction that, in any $G_i$, if $b$ flows into $a$ then $b\in \PointsToSet{a}$.
For the inverse direction, a similar induction establishes that for each $i$, 
if there is a $(4,i)$-bounded program $\Program$ that witnesses $b\in \PointsToSet{a}$
then $b$ flows into $a$ in $G_i$.

The desired result follows.
\end{proof}

\subsection{Proofs of \cref{sec:bounded}}\label{appsubsec:proof_bounded}

\smallskip
\lembellshapecorrectness*
\begin{proof}
The proof is by induction on $i$.
For $i=1$, the claim holds directly by construction.
Now assume that the claim holds for some $i$, and we show it holds for $i+1$.

We start with \cref{item:ind_epsilon}.
Consider any bell-shaped path $P\colon x\DPath{}y$ with $\MaxStackHeight(P)\leq 2^{i+1}-1$.
If $\MaxStackHeight(P)\leq 2^{i}-1$, the claim holds by the induction hypothesis and the edge $x_1\to y_1$ in $G_{i+1}$.
Otherwise, let $j_1, j_2$ be the first and last index of $P$ such that
$\StackHeight(P[:j_1])=\StackHeight(P[:j_2])=2^i$,
and $x'=P[j_1]$ and $y'=P[j_2]$. 
We have
\begin{align*}
\StackHeight(P[j_1:j_2])\leq  \MaxStackHeight(P) - 2^i = 2^{i+1} -1 -2^i = 2\cdot 2^{i} -1 -2^i = 2^i-1\ .
\end{align*}
Note that $P[:j_1]$ ends with a $\&$-labeled edge, while $P[j_2:]$ starts with a $*$-labeled edge.
In addition, since $P$ is bell-shaped, $P[j_1:j_2]$ is also bell-shaped.
By the induction hypothesis, we have $x'\Path y'$ in $G_i$, and by construction,
$x_2'\to y_2'$ in $G_{i+1}$.
Moreover, note that $P[:j_1]$ is monotonically increasing and $P[j_2:]$ is monotonically decreasing, thus, by the induction hypothesis,
we have $x_1\to x_3'$ and $y_3'\to y_1'$ in $G_{i}$.
By construction, we have $x_1\to x_2'$ and $y_2'\to y_1$ in $G_{i+1}$.
Thus, we have a path $x_1\to x_2'\to y_2' \to y_1$ in $G_{i+1}$, hence $x_1\Path y_1$, as desired.

We proceed with \cref{item:ind_increase}.
Consider any monotonically increasing path $P\colon x\DPath{2^{i}\&} y$ that ends with a $\&$-labeled edge.
Let $j$ be the first index of $P$ such that $\StackHeight(P[:j])=2^{i-1}$, and note that $P[:j]$ ends with a $\&$-labeled edge.
Note that $P[:j]$ and $P[j:]$ are monotonically increasing with $\StackHeight(P[:j])=\StackHeight(P[j:])=2^{i-1}$.
Let $z=P[j]$, and by the induction hypothesis, we have $x_1\Path z_3$ and $z_1\Path y_3$ in $G_i$.
By construction, we have $x_1\to z_2$ and $z_2\to y_3$ in $G_{i+1}$, and thus $x_1\Path y_3$, as desired.

Finally, \cref{item:ind_decrease} is similar to \cref{item:ind_increase} and is omitted for brevity. 
The desired result follows.
\end{proof}

\smallskip
\lempmaxima*
\begin{proof}
Clearly, $P'[h]$ cannot be a node of some monotonically decreasing path $P_{\ell}^{\downarrow}$.
This is because $\StackHeight(P[:l_{\ell-1}-1])$ is larger than the stack height of $P$ in all nodes of $P_{\ell}^{\downarrow}$.
Hence, $P'[h]$ is a node of some monotonically increasing path $P_{\ell}^{\uparrow}$.
Due to the monotonicity of $P_{\ell}^{\uparrow}$,  $P'[h]$ has to be the last node of $P_{\ell}^{\uparrow}$,
and thus the first node of $P[j_{\ell+1}:l_{\ell+1}]$. Hence, $f^{-1}(h)=j_{\ell+1}$.

The desired result follows.
\end{proof}

\smallskip
\lemtwomaxima*
\begin{proof}
We first argue that $P_{j_{\ell}+1}^{\downarrow}=\epsilon$.
Assume towards contradiction otherwise.
By maximality of the bell-shaped sub-paths of $P$, we have that $\StackHeight(P[:l_{\ell}+1])<\StackHeight(P[:l_{\ell}])$ (note that $P[:l_{\ell}+1]$ is the second node of $P_{j_{\ell}+1}^{\downarrow}$).
It suffices to argue that $P_{j_{\ell}}^{\uparrow}\neq\epsilon$, which will violate the maximality of the bell-shaped sub-path $P[j_{\ell}:l_{\ell}]$.
Indeed, if $P_{j_{\ell}}^{\uparrow}=\epsilon$, this would violate the fact that $P'$ has a maxima in $h$.
Hence $P_{j_{\ell}+1}^{\downarrow}=\epsilon$.
 
We now argue that $P_{j_{\ell}+1}^{\uparrow}=\epsilon$.
Indeed, given that $P_{j_{\ell}+1}^{\downarrow}=\epsilon$, if $P_{j_{\ell}+1}^{\uparrow}\neq \epsilon$
we would have $\StackHeight(P[:j_{\ell}])<\StackHeight(P[:j_{\ell}]+1)$, which contradicts the fact that $P'$ has a local maxima at $h$.

The desired result follows.
\end{proof}

\smallskip
\lemmaxdistancehalved*
\begin{proof}
Consider that $P'$ has a local maxima at some $h$.
By \cref{lem:p_maxima}, we have that $f^{-1}(h)=j_{\ell}$ for some $\ell\in [k]$.
By \cref{lem:two_maxima}, we have that  $P_{j_{\ell}+1}^{\downarrow}= P_{j_{\ell}+1}^{\uparrow}=\epsilon$.
Thus $\StackHeight(P[:l_{\ell}])=\StackHeight(P[:j_{\ell+1}])$.
Since $P[j_{\ell}, l_{\ell}]$ is bell-shaped, we have $\StackHeight(P[:j_{\ell}]) = \StackHeight(P[:j_{\ell+1}])$ and thus $P'$ does not have a local maxima in $f(j_{\ell}+1)$.
Hence, we can associate with the local maxima of $P'$ at $h$ the two unique local maxima of $P$ that appear in the bell-shaped paths $P[j_{\ell}, l_{\ell}]$ and $P[j_{\ell+1}, l_{\ell+1}]$.

The desired result follows.
\end{proof}

%

\subsection{Proofs of \cref{sec:ovhard}}\label{appsubsec:proof_ovhard}

\smallskip
\lemovcompleteness*
\begin{proof}
We first argue that in the solution graph $\ov{G}$, we have $b_D^{i_1}\DEdge{\epsilon} u_D^{i_2}$, which implies that $s$ flows into $t$.
We prove by induction the following statement.
Consider any $j\in [D]$.
\begin{compactenum}
\item If $j$ is odd, we have $v_{j}^{i_2}\DEdge{\epsilon} a_{j}^{i_1}$ in $\ov{G}$.
\item If $j$ is even, we have $b_j^{i_1}\DEdge{\epsilon} u_{j}^{i_2}$ in $\ov{G}$.
\end{compactenum}

For the base case, let $j=1$.
By construction, we have $a_{j}^{i_1}\DPath{\&}z$ and $z\DPath{\epsilon} u_{j}^{i_2}$.
In addition, since $x^{i_1}[j]=y^{i_2}[j]=0$, we have $a_{j}^{i_1}\DPath{2\&}z$ or $z\DPath{\epsilon} u_{j}^{i_2}$.
The above imply that $a_{j}^{i_1}$ flows into $u_{j}^{i_2}$, hence because of the statement $*u_{j}^{i_2}=v_{j}^{i_2}$,
we have $v_{j}^{i_2}\DEdge{\epsilon} u_{j}^{i_1}$ in $\ov{G}$, as required.

Now assume that the statement holds for $j-1$, and we argue that it holds for $j$.
First, assume that $j$ is odd.
By the induction hypothesis, we have $b_{j-1}^{i_1}\DEdge{\epsilon} u_{j-1}^{i_2}$ in $\ov{G}$.
By construction, we have $a_{j}^{i_1}\DPath{\&}b_{j-1}^{i_1}$ and $u_{j-1}^{i_2}\DPath{*}u_{j}^{i_2}$.
In addition, since $x^{i_1}[j]=y^{i_2}[j]=0$, we have $a_{j}^{i_1}\DPath{2\&}b_{j-1}^{i_1}$ or $u_{j-1}^{i_2}\DPath{\epsilon} u_{j}^{i_2}$.
The reasoning then is similar to the case of $j=1$.

Finally, assume that $j$ is even.
By the induction hypothesis, we have $v_{j-1}^{i_2}\DEdge{\epsilon} a_{j-1}^{i_1}$ in $\ov{G}$.
By construction, we have $u_{j}^{i_2}\DPath{\&}v_{j-1}^{i_2}$ and $a_{j-1}^{i_1}\DPath{*}a_{j}^{i_1}$.
In addition, since $x^{i_1}[j]=y{^i_2}[j]=0$, we have $u_{j}^{i_2}\DPath{2\&}v_{j-1}^{i_2}$ or $a_{j-1}^{i_1}\DPath{\epsilon}a_{j}^{i_1}$.
The above imply that $u_{j}^{i_2}$ flows into $a_{j}^{i_1}$, hence because of the statement $*a_{j}^{i_1}=b_{j}^{i_1}$, 
we have $b_{j}^{i_1}\DEdge{\epsilon} u_{j}^{i_2}$, as required.

Finally, note that our above analysis concerns $O(D)=O(\log n)$ nodes.
Hence there is a witness $\Program$ for $s\in \PointsToSet{t}$ that has length $\Otilde(1)$.

The desired result follows.
\end{proof}

\smallskip
\lemovsoundness*
\begin{proof}
First, observe that if $s\in \PointsToSet{t}$, there exist $i_1, i_2\in[n']$ such that $u_{D}^{i_2}$ is D-reachable from $b_{D}^{i_1}$.
Note that, in fact, $b_{D}^{i_1}\DEdge{\epsilon}u_{D}^{i_2}$, as $*$-labeled edges enter nodes that have no outgoing edges in $\ov{G}$.


We prove the following statement
For any $l_1, l_2\in [n']$, for any $j\in [D]$, the following hold.
\begin{compactenum}
\item If $j$ is odd and $a_{j}^{l_1}$ is D-reachable from $v_{j}^{l_2}$ in $\ov{G}$, then $\sum_{j'\leq j}x^{l_1}[j']\cdot y^{l_2}[j']=0$.
\item If $j$ is even and $u_{j}^{l_2}$ is D-reachable from $b_j^{l_1}$ in $\ov{G}$, then $\sum_{j'\leq j}x^{l_1}[j']\cdot y^{l_2}[j']=0$.
\end{compactenum}
For $j=D$, we have that $x^{l_1}$ and $y^{l_2}$ are orthogonal.
The proof is by induction on $j$.

For the base case, let $j=1$, and assume that $a_{1}^{l_1}$ is D-reachable from $v_{1}^{l_2}$, hence $a_{1}^{l_1}$ flows into $u_{1}^{l_2}$.
Note that all paths starting from $a_{1}^{l_1}$ with a $\&$-labeled edge go through $z$.
Hence, $a_{1}^{l_1}$ can only flow into $u_{1}^{l_2}$ via a path $P\colon P_1 \circ P_2$, where
$P_1\colon a_{1}^{l_1} \DPath{} z$ and $P_2\colon z \DPath{}u_{1}^{l_2}$.
Moreover, $\Label(P_1)=\&$ or $\Label(P_1)=\&\&$, and $\Label(P_2)=\epsilon$ or $\Label(P_2)=*$.
It follows easily by construction that $x^{l_1}[1]=0$ or $y^{l_2}[1]=0$, as otherwise $\Label(P_1)=\&$ and $\Label(P_2)=*$, which would contradict the fact that $a$ flows into $u_{1}^{l_2}$ via $P$.

Now assume that the statement holds for $j-1$, and we argue that it holds for $j$.
First assume that $j$ is odd.
By the induction hypothesis, we have that if $u_{j-1}^{l_2}$ is D-reachable from $b_{j-1}^{l_1}$ then $a_{j-1}^{l_1}\cdot b_{j-1}^{l_2}=0$.
Note that $a_{j}^{l_1}$ flows into $u_{j}^{l_2}$.
In addition, all paths starting from $a_{j}^{l_1}$ with a $\&$-labeled edge go through $b_{j-1}^{l_1}$, and thus we 
indeed have that $u_{j-1}^{l_2}$ is D-reachable from $b_{j-1}^{l_1}$.
The proof is similar to the base case, where $z$ is replaced by $b_{j-1}^{l_1}$.

Finally, assume that $j$ is even.
By the induction hypothesis, we have that if $a_{j-1}^{l_1}$ is D-reachable from $v_{j}^{l_2}$ then $a_{j}^{l_1}\cdot b_{j}^{l_2}=0$.
Note that $u_{j}^{l_2}$ flows into $a_{j}^{l_1}$.
In addition, all paths starting from $u_{j}^{l_2}$ with a $\&$-labeled edge go through $v_{l_2}^{j-1}$, and thus we
indeed have that $a_{j}^{l_1}$ is D-reachable from $v_{j}^{l_2}$.
The proof is similar to the previous case, where $a_{j}^{l_1}$ is replaced by $u_{j}^{l_2}$ and $b^{j-1}{l_1}$ is replaced by $v_{j}^{l_2}$.

The desired result follows.
\end{proof}

\subsection{Proofs of \cref{sec:pcomplete}}\label{appsubsec:proof_pcomplete}

\smallskip
\lempcompletelemma*
\begin{proof}
The ``if'' direction is straightforward, so we focus on the ``only if'' direction.
Assume towards contradiction that $x_i$ is D-reachable from $z_i$ but $x_i$ does not flow into $y_i$.
Observe that there are no incoming edges to $x_i$ in $G$.
Hence the only way to have $x_i$ D-reachable from $z_i$ in $\ov{G}$ is to have some $l\in[n]$ such that 
(i)~$x_i$ flows into some $y_l$, which leads to $z_l\DEdge{\epsilon} x_i$, and
(ii)~$z_l$ is D-reachable from $z_i$.
It suffices to argue that (ii) cannot hold.
Since $z_i$ does not have outgoing edges in $G$, all outgoing edges of $z_i$ in $\ov{G}$ are due to applications of type~4 statements on $*y_i=z_i$, and all such edges are to nodes that have an outgoing edge labeled with $\&$.
Observe that all such nodes only have outgoing edges labeled with $\&$ in $\ov{G}$.
Hence, every path $P$ that starts with $z_i$ has $\MaxStackHeight(P)>0$.
Thus, if $z_l$ is D-reachable from $z_i$, the witness path needs to traverse at least one edge  labeled with $*$.
Finally, observe that all paths in $\ov{G}$ that traverse such an edge end in a $y$ node.
Hence no such path exists in $\ov{G}$, a contradiction.

The desired result follows.
\end{proof}

\smallskip
\lempcompleteconstruction*
\begin{proof}
The ``if'' direction is straightforward, so we focus on the ``only if'' direction.
Assume that $x_i$ is D-reachable from $z_i$ in $\ov{G}$.
By \cref{lem:p_complete_lemma}, we have that $x_i$ flows into $y_i$ via a path $P$.
Note that for every path in $\ov{G}$ that traverses an edge labeled with $*$, all the following edges are also labeled with $*$.
We argue that the claim holds when $A_i$ is an AND gate, as the reasoning is similar for when $A_i$ is an OR gate.
Observe that $P$ has the form $P\colon P_1 \circ P_2 \circ P_3$, where
\begin{align*}
P_1\colon x_i \xrightarrow{\&} x_i^1 \xrightarrow{\&} x_i^2 \xrightarrow{\&} \cdots \xrightarrow{\&} x_i^i \xrightarrow{\&} z_j \qquad \text{and}\qquad
P_3\colon x_k \xrightarrow{*} y_i^{i-1} \xrightarrow{*} y_i^{i-2} \xrightarrow{*} \cdots \xrightarrow{*} y_i^1 \xrightarrow{*} y_i\ .
\end{align*}
In addition, $P_2$ witnesses the D-reachability of $x_k$ from $z_j$, and does not traverse any edges labeled with $*$, and thus no edges labeled with $\&$ either.
It follows that $P_2$ must contain the sub-paths $P_2^j$ and $P_2^k$ that witness the D-reachability of $x_j$ from $z_j$, and $x_k$ from $z_k$, respectively.

The desired result follows.
\end{proof}

\smallskip
\lempcompletecorrectness*
\begin{proof}
We prove by induction on the depth of the circuit that for every $i$, we have that $x_i$ is D-reachable from $z_i$ iff $A_i$ evaluates to $1$.

For the base case, we have that $A_i$ is an input.
If $A_i=1$, then, by construction we have $z_i\DEdge{\epsilon} x_i$ in $G$ and thus in $\ov{G}$.
On the other hand, if $A_i=0$, we have that $z_i$ does not have any outgoing edges in $G$.
Note that, since $A_i$ is an input, $z_i$ does not appear in type-4 statements, hence it will have as many outgoing edges in $\ov{G}$ as in $G$.
The claim then holds by the fact that $z_i$ does not have any outgoing edges in $G$.

We now proceed with the inductive case.
Let $A_j$ and $A_k$ be the inputs to $A_i$.
Assume that $A_i=A_j\land A_k$. 
By \cref{lem:pcomplete_construction}, we have that $x_i$ is D-reachable from $z_i$ iff $x_j$ and $x_k$ are D-reachable from $z_j$ and $z_k$, respectively.
By the induction hypothesis, we have that $x_j$ (resp., $x_k$) is D-reachable from $z_j$ (resp., $z_k$) iff $A_j$ (resp., $A_k$) evaluates to $1$.
Thus, $x_i$ is D-reachable from $z_i$ iff $A_i$ evaluates to $1$.
A similar analysis holds for $A_i=A_j\lor A_k$.

The desired result follows.
\end{proof}

\end{document}